\documentclass[11pt]{article}
\usepackage[T1]{fontenc}
\usepackage{amsfonts,amsmath,amsthm,amssymb,dsfont,mathtools}
\usepackage[margin=1in]{geometry}
\usepackage{fullpage}
\usepackage[colorlinks,citecolor=blue,linkcolor=blue,urlcolor=black,pagebackref]{hyperref}
\usepackage{mathdots}
\usepackage{bm,bbm}
\usepackage{url}
\usepackage{paralist}
\usepackage{enumerate}
\usepackage[normalem]{ulem}
\usepackage{xspace}
\xspaceaddexceptions{]\}}
\usepackage{comment}
\usepackage{mathtools}
\usepackage{tabu}
\usepackage{framed}
\usepackage{float,wrapfig}
\usepackage{subfigure}
\usepackage{tikz}
\usepackage[usenames,dvipsnames]{pstricks}
\usepackage[shortlabels]{enumitem}

\usepackage[algo2e,linesnumbered,boxed,ruled,vlined]{algorithm2e}

\usepackage{thm-restate}

\usepackage[capitalise]{cleveref}
\newtheorem{mytheorems}{Theorem}[section]
\newtheorem{them}[mytheorems]{Theorem}
\newtheorem{lem}[mytheorems]{Lemma}

\newtheorem{cor}[mytheorems]{Corollary}
\newtheorem{claim}[mytheorems]{Claim}
\newtheorem{definition}[mytheorems]{Definition}

\newtheorem{obs}[mytheorems]{Observation}

\renewcommand{\epsilon}{\varepsilon}
\newcommand{\eps}{\varepsilon}
\renewcommand{\Pr}{\operatorname*{\mathbf{Pr}}}
\newcommand{\Ex}{\operatorname*{\mathbf{E}}}

\newcommand{\poly}{\operatorname{\mathrm{poly}}}
\newcommand{\polylog}{\poly\log}

\newcommand{\N}{\mathbb{N}}

\newcommand{\Z}{\mathbb{Z}}
\renewcommand{\tilde}{\widetilde}

\newcommand{\per}{\operatorname{\mathsf{per}}} 

\newcommand{\lcp}{\operatorname{\mathsf{lcp}}}

\newcommand{\rot}{\operatorname{\mathsf{rot}}}
\newcommand{\dd}{\mathinner{.\,.}}
\newcommand{\sA}{\mathsf{A}}
\newcommand{\sB}{\mathsf{B}}
\newcommand{\sC}{\mathsf{C}}
\newcommand{\lce}{\mathsf{LCE}}
\newcommand{\sJ}{\mathsf{J}}
\newcommand{\seed}{{\sigma}}
\newcommand{\seedDS}{{\sigma_{\mathsf{DS}}}}
\newcommand{\seedH}{{\sigma_H}}
\newcommand{\sQ}{\mathsf{Q}}
\newcommand{\DS}{\mathsf{DS}}
\newcommand{\cT}{\mathcal{T}}

\DeclarePairedDelimiter{\ket}{\lvert}{\rangle}

\begin{document}

\title{Quantum Speed-ups for String Synchronizing Sets,\\ Longest Common Substring, and $k$-mismatch Matching}
\author{
Ce Jin\thanks{cejin@mit.edu. Partially supported by NSF Grant CCF-2129139}\\MIT \and Jakob Nogler\thanks{jnogler@ethz.ch}\\ETH Zurich 
}

\date{}

	\setcounter{page}{0} \clearpage
\maketitle
	\thispagestyle{empty}

	\begin{abstract}
		\emph{Longest Common Substring (LCS)} is an important text processing problem, which has recently been investigated in the quantum query model.
		 The decisional version of this problem, \emph{LCS with threshold $d$}, asks whether two length-$n$ input strings have a common substring of length $d$.  The two extreme cases, $d=1$ and $d=n$, correspond respectively to Element Distinctness and Unstructured Search, two fundamental problems in quantum query complexity.
		However, the intermediate case $1\ll d\ll n$ was not fully understood. 
		
		We show that the complexity of LCS with threshold $d$ smoothly interpolates between the two extreme cases up to $n^{o(1)}$ factors:
\begin{itemize}
	\item LCS with threshold $d$ has a quantum algorithm in $n^{2/3+o(1)}/d^{1/6}$ query complexity and time complexity, and requires at least $\Omega(n^{2/3}/d^{1/6})$ quantum query complexity.
\end{itemize}
Our result improves upon previous upper bounds $\tilde O(\min \{n/d^{1/2}, n^{2/3}\})$  (Le Gall and Seddighin ITCS 2022, Akmal and Jin SODA 2022), and answers an open question of Akmal and Jin.

Our main technical contribution is a quantum speed-up of the powerful \emph{String Synchronizing Set} technique introduced by Kempa and Kociumaka (STOC 2019). 
It consistently samples $n/\tau^{1-o(1)}$ synchronizing positions in the string depending on their length-$\Theta(\tau)$ contexts, and each synchronizing position can be reported by a quantum algorithm in $\tilde O(\tau^{1/2+o(1)})$ time.  Our quantum string synchronizing set also yields a near-optimal LCE data structure in the quantum setting.

As another application of our quantum string synchronizing set, we study the \emph{$k$-mismatch Matching} problem, which asks if the pattern has an occurrence in the text with at most $k$ Hamming mismatches. 
Using a structural result of Charalampopoulos, Kociumaka, and Wellnitz (FOCS 2020), we obtain: 
\begin{itemize}
	\item $k$-mismatch matching has a quantum algorithm with $k^{3/4} n^{1/2+o(1)}$ query complexity and $\tilde O(kn^{1/2})$ time complexity.  We also observe a non-matching quantum query lower bound of $\Omega(\sqrt{kn})$.
\end{itemize}
	\end{abstract}
	\newpage

\section{Introduction}

String processing is one of the first studied fields in the early days of theoretical computer science, yet most of the basic problems in this field are still actively being studied today, not only because of numerous applications in  bio-informatics and data mining, but also due to their inherent theoretical interest. 
Inspired by the power of quantum computers, recent works investigated \emph{quantum algorithms} for many fundamental string problems,  such as pattern matching \cite{rameshvinay}, longest common substring  \cite{legall, aj22}, edit distance \cite{boroujeni2021approximating},  regular language recognition \cite{DBLP:conf/focs/AaronsonGS19,DBLP:conf/mfcs/AmbainisBIKKPSS20}, minimal string rotation \cite{ying,aj22}.  
These studies are also connected to topics such as quantum query complexity \cite{DBLP:conf/focs/AaronsonGS19,DBLP:conf/mfcs/AmbainisBIKKPSS20}, quantum fine-grained complexity \cite{legall,ccc20aaronson,DBLP:conf/stacs/BuhrmanPS21,DBLP:conf/innovations/BuhrmanLPS22}, quantum walks and history-independent data structures \cite{legall,aj22,DBLP:conf/innovations/BuhrmanLPS22,DBLP:journals/siamcomp/Ambainis07}.

\paragraph*{Longest Common Substring}
The starting point of our paper is the \emph{Longest Common Substring (LCS)} problem: given two strings $S_1,S_2\in \Sigma^n$, the task is to compute $\mathrm{LCS}(S_1,S_2)$, defined as the maximum possible length $d$ of their common substring $S_1[i\dd i+d-1]=S_2[j\dd j+d-1]$.\footnote{We remark that in the literature the same acronym could also refer to the Longest Common Subsequence problem. The difference is that a subsequence is not necessarily a contiguous part of the string. In this paper, we only consider the Longest Common Substring problem.}
 The classical computational complexity of LCS is relatively well-understood: it can be solved in $O(n)$ time using suffix trees \cite{DBLP:conf/focs/Weiner73,DBLP:conf/focs/Farach97}. In the more powerful \emph{quantum query model}, where the input strings are given as a quantum black box, recent works showed that LCS can have \emph{sublinear} algorithms: the first result was given by Le Gall and Seddighin \cite{legall}, showing that LCS can be solved by a quantum algorithm in $\tilde O(n^{5/6})$ query complexity and time complexity.\footnote{Throughout this paper, $\tilde O(\cdot),\tilde \Omega(\cdot ),\tilde \Theta(\cdot)$ hide $\polylog(n)$ factors, where $n$ is the input size.} This bound was later improved to $\tilde O(n^{2/3})$ by Akmal and Jin \cite{aj22}.

 Both \cite{legall} and \cite{aj22} actually considered the decisional problem, \emph{LCS with threshold $d$}, which takes an extra parameter $1\le d \le n$ and simply asks whether $\mathrm{LCS}(S_1,S_2)\ge d$ holds. This problem is particularly interesting from a quantum query complexity perspective,
 as its two extreme cases correspond to two fundamental problems in this field: the $d=1$ case asks whether there exist $i,j$ such that $S_1[i] = S_2[j]$, i.e., the (bipartite) element distinctness problem (also known as the claw finding problem), which has quantum query complexity $\Theta(n^{2/3})$ due to the celebrated results of Ambainis \cite{DBLP:journals/siamcomp/Ambainis07} and Aaronson and Shi \cite{DBLP:journals/jacm/AaronsonS04}.
  The $d=n$ case asks to find  $i$ such that $S_1[i]\neq S_2[i]$, which is equivalent to the unstructured search problem, with well-known quantum query complexity $\Theta(n^{1/2})$ \cite{DBLP:journals/siamcomp/BennettBBV97} achieved by Grover Search \cite{DBLP:conf/stoc/Grover96}. 

A natural question arises: \emph{what is the true quantum query complexity of LCS with threshold $d$ in the intermediate case $1\ll d\ll n$?}\,\,
Le Gall and Seddighin \cite{legall} showed an upper bound of $\tilde O(\min \{ n^{2/3}\cdot \sqrt{d}, n/\sqrt{d}\})$ (which is $\tilde O(n^{5/6})$ in the worst case $d=n^{1/3}$). Akmal and Jin \cite{aj22} showed an upper bound of $\tilde O(n^{2/3})$ regardless of the threshold $d$, which is only near-optimal in the unparameterized setting. However, no matching lower bound is known when $1\ll d\ll n$.

Our first main result is a quantum algorithm for LCS with threshold $d$ whose complexity smoothly interpolates between the two extreme cases $d=1$ and $d=n$ (up to subpolynomial factors). This answers an open question of Akmal and Jin \cite{aj22}. A schematic comparison of previous bounds and our new bound is in \cref{myfigure}.

\begin{figure}
\centering    
\includegraphics[width=250pt]{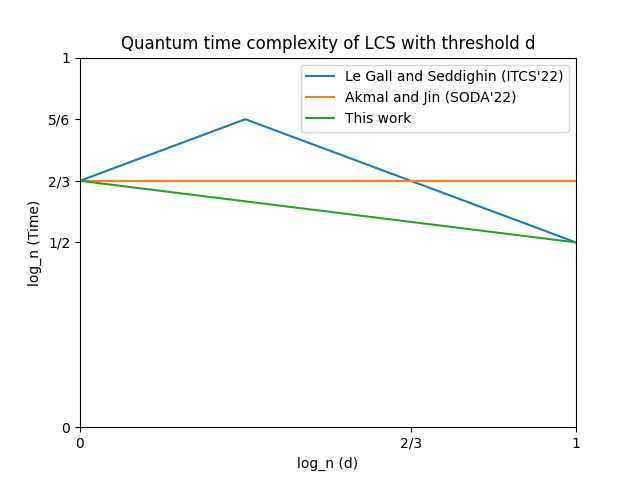}
\caption{Quantum time complexity of LCS with threshold $d$}
\label{myfigure}
\end{figure}

\begin{restatable}[LCS with threshold $d$, upper bound]{them}{lcsubrestate}
	\label{thm:lcsub}
    Given $S_1,S_2\in \Sigma^n$, there is a quantum algorithm that decides whether $S_1,S_2$ have a common substring of length $d$ in $\tilde O(n^{2/3}/d^{1/6-o(1)})$ quantum query complexity and time complexity.
\end{restatable}
We also observe a quantum query lower bound that matches the upper bound up to $d^{o(1)}$ factors. Hence, we obtain an almost complete understanding of LCS with threshold $d$ in the quantum query model.
\begin{restatable}[LCS with threshold $d$, lower bound]{them}{lcslbrestate}
	\label{thm:lcslb}
For $|\Sigma|\ge \Omega(n/d)$, deciding whether $S_1,S_2\in \Sigma^n$ have a common substring of length $d$ requires $\Omega(n^{2/3}/d^{1/6})$ quantum queries.
\end{restatable}

\paragraph*{Quantum string synchronizing sets} The main technical ingredient in our improved algorithm for  LCS with threshold $d$ (\cref{thm:lcsub}) is a quantum speed-up for constructing a \emph{String Synchronizing Set}, a powerful tool for string algorithms recently introduced by Kempa and Kociumaka \cite{kk19}. This technique was originally applied to small-alphabet strings compactly represented in word-RAM \cite{kk19},  and has later found numerous applications in other types of string problems \cite{cpm19,lcs2021,DBLP:conf/focs/KempaK20, DBLP:journals/corr/abs-2106-12725,aj22, dynamicsa}. 
   Informally speaking, for a length parameter $\tau\ge 1$ and a string $T \in \Sigma^n$ that does not have highly-periodic regions, a $\tau$-synchronizing set $\sA \subseteq [1\dd n]$ of string $T$ is a subset of positions in $T$ (called ``synchronizing positions'') that hit every length-$\Theta(\tau)$ region of $T$ at least once, and are ``consistent'' in the sense that any two identical length-$\Theta(\tau)$ substrings in $T$ should be hit at the same position(s) (in other words,  whether $i$ is included in $\sA$ only depends on the length-$\Theta(\tau)$ context around position $i$ in string $T$).
     A formal definition (with a more technical treatment of highly-periodic cases) is stated in \cref{defn:sync}, and a concrete example can be found in \cite[Figure 1]{cpm19}.
   Kempa and Kociumaka \cite{kk19} obtained a classical deterministic $O(n)$-time algorithm for constructing a $\tau$-string synchronizing set of optimal size $\Theta(n/\tau)$.

   String synchronizing sets have been naturally applied in LCS algorithms \cite{lcs2021,aj22}: for length threshold $d \gg\tau$, a length-$d$ common substring $S_1[i_1\dd i_1+d-1]=S_2[i_2\dd i_2+d-1]$ implies that these two regions contain consistently sampled synchronizing positions, i.e., there exists a shift $0\le h<d$ such that both $i_1+h$ and $i_2+h$ belong to the $\tau$-synchronizing set $\sA$ (this argument ignores the highly-periodic case, which can be dealt with otherwise). This allows us to focus on length-$d$ common substrings ``anchored at'' synchronizing positions in $\sA$, which could save computation if $\sA$ is sparse.

 To implement this anchoring idea, Akmal and Jin \cite{aj22} designed a sublinear quantum walk algorithm that finds an anchored length-$d$ common substring, assuming efficient local access to elements in the $\tau$-synchronizing set $\sA$.  
 However, this assumption is difficult to achieve, due to the high computational cost of $\sA$: in Kempa and Kociumaka's construction \cite{kk19}, elements in $\sA$ can only be reported in $O(\tau)$ (classical) time per element, which is much slower than required in Akmal and Jin's framework. To bypass this issue, Akmal and Jin composed Kempa and Kociumaka's construction \cite{kk19} with another input-oblivious construction called difference cover \cite{DBLP:conf/cpm/BurkhardtK03,DBLP:journals/tocs/Maekawa85} to reduce the reporting time, at the cost of greatly increasing the size of $\sA$. 

In this work, we directly resolve this issue faced by \cite{aj22}, by constructing a $\tau$-synchronizing set equipped with a faster quantum algorithm for reporting its elements. The sparsity of our construction is worse than optimal by only a $\tau^{o(1)}$ factor.
\begin{them}[Quantum String Synchronizing Set, informal]
Given string $T[1\dd n]$ and integer $1\le \tau\le n/2$, there is a (randomized) $\tau$-synchronizing set $\sA$ of size $n/\tau^{1-o(1)}$, such that each element of $\sA$ can be reported in $\tilde O(\tau^{\frac{1}{2}+o(1)})$ quantum query complexity and time complexity.
\label{thm:sync:infor}
\end{them}
In comparison, all previous constructions of $\tau$-synchronizing sets use at least $\tilde \Omega(\tau)$ (classical) time for reporting one element, with no quantum speed-ups known.
A formal version of \cref{thm:sync:infor} will be given in \cref{thm:quantum-sync}, with more detailed specification on the sparsity and efficient reporting.  

By plugging our quantum string synchronizing set (\cref{thm:sync:infor}) into Akmal and Jin's quantum walk algorithm \cite{aj22} for LCS with threshold $d$, we improve their quantum query and time complexity  from $\tilde O(n^{2/3})$ to $\tilde O(n^{2/3}/d^{1/6-o(1)})$, hereby obtaining \cref{thm:lcsub}.

Given numerous known applications of string synchronizing sets in classical string algorithms \cite{kk19,cpm19,lcs2021,DBLP:conf/focs/KempaK20, DBLP:journals/corr/abs-2106-12725, dynamicsa}, we expect \cref{thm:sync:infor} to be a very useful tool  in designing quantum string algorithms. As a preliminary example, we observe that Kempa and Kociumaka's optimal LCE data structure \cite{kk19} can be adapted to the quantum setting using \cref{thm:sync:infor}.

\paragraph*{An LCE data structure with quantum speed-up}
In the LCE problem, one is given a string $T\in \Sigma^n$ and needs to preprocess it into a data structure $D$, so that later one can efficiently answer $\lce(i,j)$ given any $i,j\in [n]$, defined as the length of the longest common prefix between  $T[i\dd n]$ and $T[j\dd n]$. 

We show the following result by combining \cref{thm:sync:infor} and \cite{kk19}.

\begin{restatable}[Data structure for LCE queries]{them}{lcerestate}
Given string $T\in \Sigma^n$ and integer $1\le \tau \le n/2$, there is a quantum preprocessing algorithm that outputs in $\mathcal{T}_{\mathsf{prep}}=\tilde O(n/\tau^{\frac{1}{2}-o(1)})$ time a classical data structure $D$ (with high success probability), such that:  given any $i,j\in [n]$, one can compute $\lce(i,j)$ in $\mathcal{T}_{\mathsf{ans}}=\tilde O(\sqrt{\tau})$ quantum time, given access to $T$ and $D$.
\label{thm:lce}
\end{restatable}

To put \cref{thm:lce} in context, we remark that: (1) Without any preprocessing, an LCE query can be answered in $\mathcal{T}_{\mathsf{ans}} =\tilde O(\sqrt{n})$ quantum time, using binary search and Grover search. (2) In the classical setting, it is well known that LCE data structures can be constructed in $\mathcal{T}_{\mathsf{prep}} =O(n)$ time, supporting $\mathcal{T}_{\mathsf{ans}} = O(1)$ time per LCE query (see e.g., \cite[Section 2.1]{kk19}). In comparison, \cref{thm:lce} shows that in the quantum setting there can be a trade-off of $\mathcal{T}_{\mathsf{prep}} \cdot  \mathcal{T}_{\mathsf{ans}} \le n^{1+o(1)}$. Then, we show that this trade-off is near-optimal in almost all regimes of interest.
\begin{restatable}{them}{optrestate}
	\label{cor:opt}
	An LCE data structure as described in \cref{thm:lce} must have 
	\[(\mathcal{T}_{\mathsf{prep}} + \sqrt{n} )\cdot  (\mathcal{T}_{\mathsf{ans}} + 1) \ge \tilde \Omega(n),\]
	even if $\mathcal{T}_{\mathsf{prep}}$ and $\mathcal{T}_{\mathsf{ans}}$ only measure quantum query complexity.
\end{restatable}

The LCE data structure \cref{thm:lce} will be useful in our quantum algorithm for the $k$-mismatch string matching problem.

\paragraph*{$k$-mismatch string matching}
The \textit{string matching with $k$ mismatches problem}, also referred to as the \textit{$k$-mismatch matching problem}, is that of determining whether some substring of a text $T$ has Hamming distance at most $k$ from a pattern $P$. The concern of solving efficiently this problem goes far beyond theoretical interests. Multimedia, digital libraries, and computational biology widely employ algorithms for solving this problem, since a broader and less rigid concept of string equality is needed there. 

The $k$-mismatch problem in the classical setting has been studied intensively since the 1980s. Since then, several works were published \cite{LANDAU1986239, GG86, AALMPE00, CFPSS16, GPUP17, BKW19, ckw20}, constantly improving the running time bounds for the general case and for some specific ranges for the value of $k$.  However, this problem has not been studied in the quantum computational setting so far. Many difficulties arise when trying to adapt earlier classical string algorithms to the quantum case. For instance, in these algorithms FFT is heavily employed, which hardly adapts efficiently enough to the quantum setting. 
Fortunately, recent studies in string matching with $k$ mismatches have discovered several useful structural results.
   Bringmann, K\"unnemann and Wellnitz \cite{BKW19} brought up the fundamental question of what the solution structure of pattern matching with $k$ mismatches would look like. Given a pattern $P$ of length $m$ and a text of length $2m$, they found out that either the number of $k$-mismatch occurrences of $P$ in $T$ is bounded by $O(k^2)$, or $P$ is approximately periodic, meaning that there exists a string $Q$ such that the number of mismatches between $P$ and the periodic extension of $Q$ is $O(k)$. 
Charalampopoulus, Kociumaka and Wellnitz \cite{ckw20} improved the bound on the number of occurrences of matches in the non approximately periodic case from $O(k^2)$ to $O(k)$, and showed that it is tight. They offered a constructive proof and provided a meta-algorithm that can be adapted with appropriate subroutines to the quantum setting.

For the $k$-mismatch matching problem our first result is the following: 
\begin{them} 
\label{thm:mismatch}
Let $P \in \Sigma^m$, let $T \in \Sigma^n$, and let $k \in [1\dd m]$ be a threshold. Then, we can verify the existence of a $k$-mismatch occurrence of $P$ in $T$ (and report its starting position in case it exists) in $\tilde{O}(k\sqrt{n})$ query complexity and time complexity.
\end{them}
This algorithm can be modified to produce a new one having a slightly better query complexity (up to an $n^{o(1)}$ factor). 
\begin{them} 
\label{thm:querymismatch}
Let $P \in \Sigma^m$, let $T \in \Sigma^n$, and let $k \in [1\dd m]$ be a threshold. Then, we can verify the existence of a $k$-mismatch occurrence of $P$ in $T$ (and report its starting position in case it exists) in $\tilde O(k^{3/4}n^{1/2}m^{o(1)})$ query complexity and $\tilde{O}(k\sqrt{n})$ time complexity.
\end{them}

We remark that $k$-mismatch matching is no easier than the problem of checking whether $T\in \{0,1\}^n$ has Hamming weight at most $k$, which  has a quantum query lower bound of $\Omega(\sqrt{kn})$ for small $k$ \cite{polynomial,Paturi92}. We leave it as an open question to close the gap between this lower bound and our upper bound for $k$-mismatch matching.

\subsection{Technical Overview}
\paragraph*{String synchronizing sets}
We outline our construction of $\tau$-synchronizing sets  with $\tilde O(\tau^{1/2+o(1)})$ quantum reporting time per element (\cref{thm:sync:infor}). For simplicity, here we only consider the non-periodic case, i.e., we assume that the input string $T[1\dd n]$ does not contain any length-$\tau$ substring with period at most $\tau/3$.

In this non-periodic case, our starting point is a simple randomized construction of \cite{kk19}, which is much simpler than their deterministic construction (whose sequential nature makes it more difficult to have a quantum speed-up).
Their idea is to \emph{pick local minimizers of a random hash function}. More specifically, sample a random hash value $\phi(S)$ for every $S\in \Sigma^{\tau}$,   and denote $\Phi(i):= \phi(T[i\dd i+\tau))$. Then, include $i$ in the synchronizing set if and only if $\min_{j\in [i\dd i+\tau]} \Phi(j)$ is achieved at $j\in \{i, i+\tau\}$. It is straightforward to verify that, (1) whether $i\in \sA$ is completely determined by $T[i\dd i+2\tau)$ and the randomness, and (2) every length-$\tau$ interval contains at least one $i\in \sA$. So $\sA$ is indeed an $\tau$-synchronizing set.
Then, the non-periodic assumption ensures that nearby length-$\tau$ substrings of $T$ are distinct, so that the probability of $i\in \sA$ is $1/\Omega(\tau)$ for every $i$, which implies the sparsity of $\sA$ in  expectation.

To implement the above idea in the quantum setting, the first challenge is to implement a hash function $\phi$ on length-$\tau$ substrings that can be evaluated in at most $\tilde O(\tau^{1/2+o(1)})$ quantum time.  
Naturally, the hash function should be random enough. A minimal (but not sufficient) requirement seems to be that $\phi$ should at least be able to distinguish two different strings $x,y\in \Sigma^\tau$ by outputting different values $\phi(x)\neq \phi(y)$ with good probability. However, it is not clear how such hash family can be implemented in only $\tilde O(\tau^{1/2+o(1)})$ quantum time. 
Many standard hash functions that have this property,  such as Karp-Rabin fingerprints, provably require at least $\Omega (\tau)$ query complexity.

To overcome this challenge, we observe that it is not necessary to use a hash family with full distinguishing ability.  In order for the randomized construction to work, we only need that the hash values of two heavily overlapping length-$\tau$ substrings to behave like random.
  Fortunately, this weaker requirement can be satisfied using the \emph{deterministic sampling} method of Vishkin \cite{DBLP:journals/siamcomp/Vishkin91}.  This technique was originally used for parallel algorithms for exact string matching \cite{DBLP:journals/siamcomp/Vishkin91}, and was later adapted into a quantum algorithm for exact string matching in $\tilde O(\sqrt{n})$ time by Ramesh and Vinay \cite{rameshvinay}, as well as some other types of exact string matching problem, e.g., \cite{DBLP:conf/cpm/GasieniecPR95,DBLP:conf/stacs/CrochemoreGPR95}. In the context of string matching, the idea is to carefully sample $O(\log n)$ positions in the pattern, so that a candidate match in the text that agrees on all the sampled positions can be used to rule out other nearby candidate matches, and hence save computation by reducing the number of candidate matches that have to undergo a linear-time full check against the pattern.
   In our situation, we use the quantum algorithm for deterministic sampling \cite{rameshvinay} in $\tilde O(\sqrt{\tau})$ time, and use these carefully sampled positions to build a hash function that is guaranteed to evaluate to different values on two heavily overlapping strings.
   To the best of our knowledge, this is the first application of deterministic sampling in a completely different context than designing string matching algorithms.

Having designed a suitable hash function $\phi$, a straightforward attempt to report a synchronizing position is to use the quantum minimum finding algorithm \cite{minimumfinding} on a length-$\tau$ interval $[j\dd j+\tau)$ to find the $i\in [j\dd j+\tau)$ with minimal hash value $\Phi(i)=\phi(T[i\dd i+\tau))$. This incurs $O(\sqrt{\tau})$ evaluations of the hash function, each taking $\tilde O(\sqrt{\tau})$ quantum time, which would still be $\tilde O(\tau)$ in total, slower than our goal of $\tilde O(\tau^{1/2+o(1)})$.

To obtain better quantum query and time complexity, we further modify our construction of the hash function $\phi$, so that  one can find the minimal hash value in a length-$\tau$ interval in a tournament tree-like fashion: The leaves are the $\tau$ candidates, each internal node picks the minimum hash value among its children using quantum minimum finding, and the root will be the minimum among all $\tau $ intervals. Here, the crucial point is to make sure that comparing two nodes in a lower level (corresponding to two closer candidates) can take less time.
 This makes sure that the total quantum time complexity of finding the minimum hash value is still $\tau^{1/2+o(1)}$.
We remark that this tournament tree structure of our algorithm is inspired by the Lexicographically Minimal String Rotation algorithm by Akmal and Jin \cite{aj22}, but our construction requires more technical details to deal with the case where a node in the tournament tree returns multiple minimizers. 

\paragraph*{$k$-mismatch matching}

Our work heavily draws from the structural insights into pattern matching with mismatches described by Charalampopoulus, Kociumaka and Wellnitz \cite{ckw20}, where an algorithm is given that analyzes the pattern $P$ for useful structure that can help bounding the number of $k$-mismatch occurrences. The algorithm constructs either a set $\mathcal{B}$ of $2k$ breaks, a set $\mathcal{R}$ of disjoint repetitive regions, or determines that $P$ is almost periodic. We adapt the algorithm to the quantum case, obtaining a subroutine taking $\tilde{O}(\sqrt{km})$ time. In order to find efficiently a $k$-mismatch occurrence of $P$ in $T$, we behave differently depending on the outcome of the subroutine. 

If the analysis resulted in breaks or repetitive regions, it is possible to identify $O(nk/m)$ candidate positions for finding a $k$-mismatch occurrence of $P$ in $T$. Using Grover search, we verify each of the candidate positions by finding one after another $k$ mismatches between the pattern and the text. If we do this in a naive way, this costs $\tilde{O}(k\sqrt{m})$ time. We show how to combine Grover search with an exponential search such that the required time can be reduced to $\tilde{O}(\sqrt{km})$. Intuitively, by using again Grover search we can verify all candidate positions in $\tilde{O}(\sqrt{nk/m} \cdot \sqrt{km}) = \tilde{O}(k\sqrt{n})$ time. The time analysis is however a little bit more tricky: finding the $O(nk/m)$ candidate positions involves having to work with subroutines which might take longer in case we found a candidate position and thus have to verify it. To overcome this issue, we have to use quantum search with variable times \cite{AA08, DBLP:conf/mfcs/CornelissenJOP20}.

Otherwise, we determined that the pattern $P$ has an approximate period $Q$, meaning that there are only $O(k)$ mismatches between $P$ and the periodic extension of $Q$. If there exists a $k$-mismatch occurrence $P$ in $T$, then there exists a long enough region in $T$ with few mismatches with a periodic extension of $Q$ as well. The key insight here is that all $k$-mismatch occurrences overlapping with this region have to align with the approximate period $Q$. This allows us to compare only the $O(k)$ positions in $P$ and $T$ which differ from $Q$ for each possible starting position of a $k$-mismatch occurrence. As there are at most $O(n)$ of them, we can use Grover's search over them to find a $k$-mismatch of $P$ in $T$ in $O(k\sqrt{n})$.

Lastly, we observe that, using the LCE data structure based on our quantum string synchronizing sets, we can use kangaroo jumping to verify a candidate position in $\tilde{O}(k^{1/4}\sqrt{m})$ query complexity, improving upon the prior verification algorithm. This leads to an algorithm for the $k$-mismatch problem requiring $\tilde{O}(k^{3/4}n^{1/2}m^{o(1)})$ query complexity and $\tilde{O}(k\sqrt{n})$ time complexity.

\subsection{Paper Organization}
In \cref{sec:prelim} we introduce basic definitions and useful lemmas that will be used throughout the paper.
 In \cref{sec:sync} we prove our main technical result, a string synchronizing set with quantum speedup.
Then in the next two sections we describe consequences of our quantum string synchronizing set combined with previous techniques: in \cref{sec:lcs} we describe the improved quantum algorithm for LCS with threshold $d$, and show a nearly matching lower bound. In \cref{sec:lce}, we describe the data structure for answering LCE queries, and show a nearly matching lower bound. Finally, in \cref{sec:kmis}, we describe our algorithm for the $k$-mismatch matching problem.
We conclude the paper with open questions in \cref{sec:open}.
\section{Preliminaries}
\label{sec:prelim}

\subsection{Notations and Basic Properties of Strings}
\label{sec:prelim1}
Throughout this paper, $\tilde O(\cdot),\tilde \Omega(\cdot ),\tilde \Theta(\cdot)$ hide $\polylog(n)$ factors, where $n$ is the input size.

Define sets $\N = \{0,1,2,3,\dots\}$, $\N^+ = \{1,2,3,\dots\}$ and $\Z = \{\dots,-1,0,1,\dots\}$. For every positive integer $n$, let $[n]= \{1,2,\dots,n\}$.  For integers $i\le j$,  let $[i\dd j]= \{i,i+1,\dots,j\}$ denote the set of integers in the closed interval $[i,j]$.   We define $[i\dd j),(i\dd j]$, and $(i\dd j)$ analogously.

We consider strings over a \emph{polynomially-bounded integer alphabet} $\Sigma = [1\dd n^{O(1)}]$.
A string $S \in \Sigma^n$ is a sequence of characters $S=S[1]S[2]\cdots S[n]$ from the alphabet $\Sigma$ (we use 1-based indexing). 
The \emph{concatenation} of two strings $S,T\in \Sigma^*$ is denoted by $ST$. Additionally, we set $S^k := S \cdots S$ for $k \in \N^+$ to be the concatenation of $k$ copies of $S$, and we denote with $S^*$ the infinite string obtained by concatenating infinitely many copies of $S$. We call a string $T$ \emph{primitive} if it cannot be expressed as $T = S^k$ for a string $S$ and an integer $k > 1$.

Given a string $S$ of length $|S|=n$, a \emph{substring} of $S$ is any string of the form $S[i\dd j] = S[i]S[i+1]\cdots S[j]$ for some indices $1\le i\le j\le n$. 
We sometimes use $S[i\dd j) = S[i]S[i+1]\cdots S[j-1]$ and $S(i \dd j]=S[i+1]\cdots S[j-1]S[j]$ to denote substrings. 
A substring $S[1\dd j]$ is called a \emph{prefix} of $S$, and a substring $S[i\dd n]$ is called a \emph{suffix} of $S$. 
For two strings $S,T$, let $\lcp(S,T)= \max\{j: j\le \min\{|S|,|T|\}, S[1\dd j] = T[1\dd j] \}$ denote the length of their \emph{longest common prefix}.

For a positive integer $p\le |S|$, we say $p$ is a \emph{period} of $S$ if $S[i] = S[i+p]$ holds for all $1\le i\le  |S|-p$. 
We refer to the minimal period of $S$ as \emph{the period} of $S$, and denote it by $\per(S)$.  
If $\per(S)\le |S|/2$, we say that $S$ is \emph{periodic}. If $\per(S)$ does not divide $|S|$, we say that $S$ is \emph{primitive}.
A \emph{run} in $T$ is a periodic substring that cannot be extended (to the left nor to the right) without an increase of its shortest period.

We need the following well-known lemmas about periodicity in strings.
\begin{lem}[Weak Periodicity Lemma, {\cite{periodicity}}]
\label{lem:weak-period}
If a string $S$ has periods $p$ and $q$ such
that $p + q \le |S|$, then $\gcd(p, q)$ is also a period of $S$.
\end{lem}

\begin{lem}[Structure of substring occurrences, e.g., \cite{DBLP:conf/icalp/PlandowskiR98,ipm}]
\label{lm:period-patterns}
Let $S,T$ be two strings with $2|T|/3 \le |S|\le |T|$, and let $T[k_1\dd k_1+|S|)= T[k_2\dd k_2+|S|)= \dots =  T[k_d\dd k_d+|S|) = S$ be all the occurrences of $S$ in $T$ (where $k_j<k_{j+1}$ for $1\le j<d$). Then, $k_1,k_2,\dots,k_d$ form an arithmetic progression. Moreover, if $d\ge 2$, then $\per(S) =k_2-k_1$.
\end{lem}

For a string $S\in \Sigma^n$, define the \emph{rotation} operations: $\rot(S)=S[n]S[1\dd n-1]$, and $\rot^{-1}(S)=S[2\dd n]S[1]$. 
For integer $j \ge 0$, $\rot^{j}(S)$ is the $\rot$ operation executed $j$ times on $S$. Similarly, for $j < 0$, $\rot^{j}(S)$ is the $\rot^{-1}$ operation executed $|j|$ times on $S$.

We say string $S$ is \emph{lexicographically smaller} than string $T$ (denoted $S \prec T$) if either $S$ is a proper prefix of $T$ (i.e., $|S|<|T|$ and $S=T[1\dd |S|]$), or $\ell = \lcp(S,T)<\min\{|S|,|T|\}$ and $S[\ell+1]<T[\ell+1]$. The notations $\succ, \preceq,\succeq$ are defined analogously. 

For a periodic string $S$ with shortest period $\per(S)=p$, the \emph{Lyndon root} of $S$ is defined as the lexicographically minimal rotation of $S[1\dd p]$.

Given two strings $S$ and $T$ with $|S| = |T|$, the set of \emph{mismatches} between $S$ and $T$ is defined as $\text{Mis}(S,T) := \{i \in [n] : S[i] \neq T[i]\}$, and the  \emph{Hamming distance} between $S$ and $T$ is denoted as $\delta_H(S,T) := |\text{Mis}(S,T)|$.
For convenience, given a string $S$ and an infinite string $Q^*$ we will often write $\delta_H(S,Q^*)$ instead of $\delta_H(S,Q^*[1\dd|S|])$.

\subsection{Computational Model}
\label{sec:model}
We assume the input string $S\in \Sigma^n$ can be accessed in a quantum query model \cite{ambainis2004quantum,DBLP:journals/tcs/BuhrmanW02}: there is an input oracle $O_S$ that performs the unitary mapping $O_s\colon \ket {i,b}  \mapsto \ket{i,b\oplus S[i]}$ for any index $i\in [n]$ and any $b\in \Sigma$.

The \emph{query complexity} of a quantum algorithm (with $2/3$ success probability) is the number of queries it makes to the input oracles. The \emph{time complexity} of the quantum algorithm additionally counts the elementary gates \cite{PhysRevA.52.3457} for implementing the unitary operators that are independent of the input. Similar to prior works \cite{legall,aj22,DBLP:journals/siamcomp/Ambainis07}, we assume quantum random access quantum memory, where the memory consists of an array of qubits, and random access can be made in superposition.

We say an algorithm succeeds \emph{with high probability (w.h.p)}, if the success probability can be made at least $1-1/n^c$ for any desired constant $c>1$. A bounded-error algorithm can be boosted to succeed w.h.p.\ by $O(\log n)$ repetitions.
In this paper, we do not optimize the log-factors of the quantum query complexity (and time complexity) of our algorithms.

\subsection{Basic Quantum Primitives}
\label{sec:primitive}

\paragraph*{Grover search (Amplitude amplification) \cite{DBLP:conf/stoc/Grover96,brassard2002quantum}}
Let $f\colon [n] \to \{0,1\}$ be a function, where $f(i)$ for each $i\in [n]$ can be evaluated in time $T$. There is a quantum algorithm that, with high probability, finds an $x\in f^{-1}(1)$ or reports that $f^{-1}(1)$ is empty, in $\tilde O(\sqrt{n}\cdot T)$ time.  Moreover, if it is guaranteed that either $|f^{-1}(1)| \ge M$ or $|f^{-1}(1)|=0$ holds, then the algorithm runs in $\tilde O(\sqrt{n/M} \cdot T)$ time.

\paragraph*{Quantum search with variable times \cite{AA08, DBLP:conf/mfcs/CornelissenJOP20}} Let $f\colon [n] \to \{0,1\}$ be a function, where $f(i)$ for each $i\in [n]$ can be evaluated in time $t_i$. There is a quantum algorithm that, with high probability, finds an $x\in f^{-1}(1)$ or reports that $f^{-1}(1)$ is empty, in $\tilde O (\sqrt{\sum_{i=1}^{n} t_i^2})$ time. 

\paragraph*{Quantum minimum finding \cite{minimumfinding}} Let $x_1,\dots,x_n$ be $n$ items with a total order, where each pair of $x_i$ and $x_j$ can be compared in time $T$. There is a quantum algorithm that, with high probability, finds the minimum item among $x_1,\dots,x_n$ in $\tilde O(\sqrt{n}) \cdot T$ time.

\subsection{Quantum Algorithms on Strings}

The following algorithm follows from a simple binary search composed with Grover search.
\begin{obs}[Finding Longest Common Prefix]
    \label{obs:lcp}
    Given $S,T\in \Sigma^n$, there is an $\tilde O(\sqrt{n})$-time quantum algorithm that computes $\lcp(S,T)$.
\end{obs}

Ramesh and Vinay \cite{rameshvinay} obtained a near-optimal quantum algorithm for exact pattern matching. 

\begin{restatable}[Exact Pattern Matching \cite{rameshvinay}]{them}{thmmatchingrestate}
    Given pattern $S\in \Sigma^m$ and text $T\in \Sigma^n$ with $n\ge m$, there is an $\tilde O(\sqrt{n})$-time quantum algorithm that finds an occurrence of $S$ in $T$ (or reports none exists).
    \label{thm:matching}
\end{restatable}

Kociumaka, Radoszewski, Rytter, and Wale\'{n} \cite{ipm} showed that computing $\per(S)$ can be reduced to $O(\log |S|)$ instances of exact pattern matching and longest common prefix involving substrings of $S$.  Hence, we have the following corollary. See also \cite[Appendix D]{ying}.
\begin{cor}[Finding Period]
    \label{cor:per}
   Given $S\in \Sigma^n$, there is a quantum algorithm in $\tilde O(\sqrt{n})$ time that computes $\per(S)$.
\end{cor}

\subsection{Pseudorandomness}

We will use the following min-wise independent hash family.

\begin{lem}[Approximate min-wise independent hash family, follows from \cite{DBLP:journals/jal/Indyk01}]
    \label{lem:minwise}
Given integer parameters $n \ge 1$ and $N\ge c\cdot n^3$ (for some constant $c$), there is a hash family $\mathcal{H} = \{h\colon [N]\to [N^2]\}$, where each  $h\in \mathcal{H}$ is an \emph{injective} function that can be specified using $O(\log N)$ bits, and can be evaluated at any point in $\polylog(N)$ time.\footnote{The original definition in \cite{DBLP:journals/jal/Indyk01} used functions from $[N]\to [N]$ and does not guarantee injectivity with probability $1$. Here we can guarantee injectivity by simply attaching the input string to the output. Doing this will slightly simplify some of the presentation later.}

Moreover, $\mathcal{H}$ satisfies \emph{approximate min-wise independence}:
     for any $x\in [N]$ and subset $X\subseteq [N] \setminus \{x\}$ of size $|X|\le n$, over a uniformly sampled $h\in \mathcal{H}$,
\[ \Pr_{h\in \mathcal{H}} \big [h(x) < \min \{h(x'): x' \in X\}\big ] \in \frac{1}{|X|+1}\cdot (1\pm 0.1).\]
\end{lem}

\section{String Synchronizing Sets with Quantum Speedup}
\label{sec:sync}

We first state the definition of string synchronizing sets introduced by Kempa and Kociumaka \cite{kk19}.

\begin{restatable}[String synchronizing set \cite{kk19}]{definition}{defnsyncrestate}
For a string $T[1\dd n]$ and a positive integer $1\le \tau \le n/2$, we say $\sA\subseteq [1\dd n-2\tau+1]$ is a \emph{$\tau$-synchronizing set of $T$} if it satisfies the following properties:
\begin{itemize}
    \item \textbf{Consistency:} If $T[i\dd i+2\tau) = T[j\dd j+2\tau)$, then $i\in \sA$ if and only if $j\in \sA$. 
    \item \textbf{Density:} For $i\in [1\dd n-3\tau +2]$, $\sA \cap [i\dd i+\tau) =\emptyset$ if and only if $\per(T[i\dd i+3\tau - 2])\le \tau/3$.
\end{itemize}
\label{defn:sync}
\end{restatable}

A concrete example of a string synchronizing set can be found in \cite[Figure 1]{cpm19}.
Note that the periodicity condition in the density requirement is necessary for allowing a   small $\sA$ to exist; otherwise, in a string with very short period, one would have to include too many positions in $\sA$ due to the consistency requirement.

Our main technical result is the following theorem.

\begin{them}[Quantum string synchronizing set]
\label{thm:quantum-sync}
Given string $T[1\dd n]$ and integer $1\le \tau \le n/2$, there is a randomized $\sA \subseteq [1\dd n-2\tau+1]$ generated from random seed $\seed\in\{0,1\}^{\polylog (n)}$, with the following properties:
 \begin{itemize}
    \item \textbf{Correctness:} $\sA$ is always a $\tau$-synchronizing set of $T$.
    \item \textbf{Sparsity:}  
    For every $i\in [1\dd n-3\tau+2]$,   $\Ex_{\seed}\big [|\sA\cap [i\dd i+\tau)|\big ] \le \tau^{o(1)}$.
    \item \textbf{Efficient computability:}
   With high probability over $\sigma$, there is quantum algorithm that,  given $i\in [1\dd n-3\tau+2]$, $\sigma$, and quantum query access to $T[i\dd i+3\tau-2]$, reports all the elements in $\sA\cap [i\dd i+\tau)$ in 
    \[ \left (\mathsf{cnt}+1 \right ) \cdot \tau^{\frac{1}{2}+o(1)}\cdot \polylog(n)  \] quantum time, where $\mathsf{cnt}=|\sA\cap [i\dd i+\tau)|$ is the output count.
\end{itemize}
\end{them}

\subsection{Review of Kempa and Kociumaka's construction}
\label{sec:defn}

\newcommand{\im}{{\Phi}}
We first review the underlying framework of Kempa and Kociumaka's construction \cite{kk19}, which will be used in our construction as well.  Define set
\begin{equation}
    \label{eqn:defnQ}
    \sQ := \{i \in [1\dd n-\tau+1]:  \per(T[i\dd i+\tau))\le \tau/3\},
\end{equation}
which contains the starting positions of highly-periodic length-$\tau$ substrings.

Let $\phi \colon \Sigma^{\tau} \to X$ be a mapping from length-$\tau$ strings to some totally ordered set $X$. 
For index $i\in [1\dd n-\tau+1]$ in the input string $T[1\dd n]$, we introduce the shorthand 
\begin{equation}
 \im(i) := \phi(T[i\dd i+\tau) ).
\end{equation}

Then, define the synchronizing set as
\begin{equation}
    \sA := \big \{i \in [1\dd n-2\tau+1]: \min \{\im(j): j\in [i\dd i+\tau]\setminus \sQ \} \in \{\im(i),\im(i+\tau)\} \big \}.
    \label{eqn:defnS}
\end{equation}
Then, Kempa and Kociumaka proved the following lemma.
\begin{lem}[{\cite[Lemma 8.2]{kk19}}]
    \label{lem:kk}
   The set $\sA$ defined in \eqref{eqn:defnS}  is always a $\tau$-synchronizing set of $T$. 
\end{lem}

It is straightforward to check that $\sA$ satisfies the consistency property in \cref{defn:sync}.
To provide an intuition, we briefly restate their proof  for the density property in the simpler case where $\sQ = \emptyset$. Given $i$, pick $k \in [i\dd i+2\tau)$ with the minimum $\Phi(k)$. If $k<i+\tau$, then $\Phi(k) = \min \{ \Phi(j) : j\in [k\dd k+\tau] \}$, and hence $k\in \sA$. Otherwise,  $\Phi(k) = \min \{ \Phi(j) : j\in [k-\tau \dd k] \}$, and hence $k-\tau\in \sA$. In either case, $\sA \cap [i\dd i+\tau) \neq \emptyset$.

From now on, we assume $\tau \ge 100$. 
For $\tau<100$, one can take $\phi$ to be the identity mapping and use the definition in \eqref{eqn:defnS}, and the requirements in \cref{thm:quantum-sync} are automatically satisfied by a brute force algorithm.

The following property on the structure of set $\sQ$ will be needed later.
\begin{lem}
    \label{lem:structureQ}
    Suppose interval $[l \dd r] \subseteq [1\dd n-\tau+1]$ has length $r-l+1\le  \tau/3$. Then,  $[l\dd r]\cap \sQ$ is either empty or an interval.
\end{lem}
\begin{proof}
    We only need to consider the case where $|[l\dd r]\cap \sQ| \ge 2$, and let $j<k$ be the minimum and maximum element in $[l\dd r]\cap \sQ$, respectively.
      By definition of $\sQ$, we have $p=\per(T[j\dd j+\tau-1])\le \tau/3$ and $q=\per(T[k\dd k+\tau-1]) \le \tau/3$. Applying weak periodicity lemma (\cref{lem:weak-period}) on the substring $T[k\dd j+\tau-1]$ of length $j+\tau-k\ge 2\tau/3$, we know $\gcd(p,q)$ is a period of $T[k\dd j+\tau-1]$. So  $\gcd(p,q)$ is also a period of $T[j\dd k+\tau-1]$. This implies $[j\dd k] \subseteq \sQ$, and hence $[l\dd r] \cap \sQ$ equals the contiguous interval $[j\dd k]$.
\end{proof}

In the following sections, our goal is to design a suitable (randomized) mapping $\phi$ that makes $\sA$ sparse in expectation and allows an efficient quantum algorithm to report elements in $\sA$. 
Our main tool for constructing $\phi$ is Vishkin's Deterministic Sampling, which we shall review in the following section.

\subsection{Review of Vishkin's Deterministic Sampling}

Vishkin's Deterministic Sampling is a powerful technique originally designed for parallel pattern matching algorithms \cite{DBLP:journals/siamcomp/Vishkin91}.  Ramesh and Vinay \cite{rameshvinay} adapted this technique into a quantum algorithm for pattern matching (\cref{thm:matching}).

The following version of deterministic sampling is taken from Wang and Ying's presentation \cite{ying} of Ramesh and Vinay's quantum pattern matching algorithm. 
Compared to the original classical version by Vishkin, the main differences are: (1) The construction takes a random seed $\sigma$ (but it is called deterministic sampling nonetheless), and (2) it explicitly deals with periodic patterns.

\begin{restatable}[Quantum version of Vishkin's Deterministic Sampling, \cite{rameshvinay,ying}]{lem}{lemdsrestate}
   \label{lem:ds}
Let round parameter $w = O(\log n)$.  Given $S\in \Sigma^n$, every seed $ \seed \in \{0,1\}^w$ determines a \emph{deterministic sample} $\DS_{\seed}(S)$ consisting of an \emph{offset} $\delta$ and a sequence of $w$ \emph{checkpoints} $i_1,i_2, \dots , i_w$,  with the following properties:
\begin{itemize}
 \item $\DS_{\seed}(S)= (\delta; i_1,i_2,\dots,i_w)$ can be  computed in $\tilde O(\sqrt{n})$ quantum time given $S$ and $\seed$.
 \item The offset satisfies $\delta \in [0\dd \lfloor n/2\rfloor ]$.
 \item The checkpoints satisfy $i_j-\delta \in [1\dd n]$ for all $j\in [w]$.
 \item Given $S\in \Sigma^n$, the following holds with at least $1- n/2^w$ probability over uniform random $\seed \in \{0,1\}^w$: for every offset $\delta' \in  [0\dd \lfloor n/2\rfloor ]  $ with $\delta'-\delta$ not being a multiple of $\per(S)$, $\DS_{\seed}(S)$ contains a  checkpoint $i_j$ such that $i_j-\delta' \in [1\dd n]$ and $S[i_j-\delta']\neq S[i_j-\delta]$.

If the property above is satisfied, we say $\seed$ generates a \emph{successful} deterministic sample  $\DS_{\seed}(S)$ for $S$.
\end{itemize}
\end{restatable}

For completeness, we include the proof of \cref{lem:ds} in \cref{app:ds}. %

In the following sections, we will fix a randomly chosen $\seed \in \{0,1\}^w$ with a suitable length $w=O(\log n)$. 
 By a union bound, we can assume that $\DS_{\seed}(S)$ is successful for all substrings $S$ of the input string $T$.

\subsection{Construction of Mapping \texorpdfstring{$\phi$}{phi}}
To build the $\tau$-synchronizing set, our construction of the mapping $\phi$ will be an $L$-level procedure, from top to bottom with length parameters 
\[\tau = \tau_L> \tau_{L-1}>\dots > \tau_2 > \tau_{1}= 6.\] 
These parameters will be chosen later (we will choose $L=\Theta(\sqrt{\log \tau})$).
In the following, we will apply Vishkin's deterministic sampling (\cref{lem:ds}) to substrings of the input string $T[1\dd n]$ of length $m \in \{\tau_1,\tau_2,\dots,\tau_L\}$.

We first define a mapping $\pi$ by augmenting the deterministic sample with some extra information.
\begin{definition}[Mapping $\pi$]
    \label{defn:mapF}
    For $S\in \Sigma^m$ with deterministic sample $\DS_{\seed}(S) = (\delta; i_1,i_2,\dots,i_w)$, define tuple \[\pi_{\seed}(S):= (\per(S); \delta; i_1,\dots,i_w; S[i_1-\delta],\dots,S[i_w-\delta]).\]
\end{definition}

Since $w= O(\log n)$, the tuple $\pi_{\seed}(S)$ can be encoded using $\polylog(n)$ bits.
 The following key lemma says $\pi$ can distinguish two non-identical strings with large overlap.

\begin{lem}
    \label{lem:vishkin}
 Let $1\le d \le \lfloor m/4\rfloor $ and $S\in \Sigma^{m+d}$. Define two overlapping length-$m$ substrings $U = S[1\dd m], V = S[d+1\dd d+m]$. Suppose seed $\seed$ generates successful deterministic samples for both $U$ and $V$.
 
 Then, $\pi_{\seed}(U)=\pi_{\seed}(V)$ if and only if $U=V$, in which case $d$ is a period of $S$.
\end{lem}
\begin{proof} The if part is immediate, and we shall prove the only if part.

   Assuming $\pi_{\seed}(U)=\pi_{\seed}(V)$, let $\per(U)=\per(V) = p$, $\DS_{\seed}(U)=\DS_{\seed}(V) =(\delta; i_1,i_2,\dots,i_w)$ where $0\le \delta \le \lfloor m/2\rfloor$ and $i_j-\delta \in [1\dd m]$, and we have $U[i_j-\delta]=V[i_j-\delta]$ for all $j\in [w]$.  We will first show $d$ must be a multiple of $p$. 
   
   Suppose for contradiction that $d$ is not a multiple of $p$.
       Since $\max\{ \delta, \lfloor m/2 \rfloor -\delta\} \ge  \lfloor m/4\rfloor \ge d$, at least one of the following two cases must happen:
       \begin{itemize}
           \item If $d\le \delta$, then $0\le \delta - d\le \lfloor m/2\rfloor$, and the property of $\DS_{\seed}(U)$ implies the existence of a checkpoint $i_j$ with $i_j-(\delta-d) \in [1\dd m]$ such that $U[i_j-\delta] \neq U[i_j-(\delta- d)]$, contradicting $U[i_j-\delta]=V[i_j-\delta] = U[i_j-\delta+d]$.
           \item If $d\le \lfloor m/2\rfloor-\delta$, then $0\le \delta + d\le \lfloor m/2\rfloor$, and the property of $\DS_{\seed}(V)$ implies the existence of a checkpoint $i_j$ with $i_j-(\delta+d)\in [1\dd m]$ such that $V[i_j-\delta] \neq V[i_j - (\delta+d)]$, contradicting $V[i_j-\delta]=U[i_j-\delta] = V[i_j-\delta-d]$.
       \end{itemize}
       Hence, $d$ must be a multiple of $p$, and hence $p\le d\le \lfloor m/4\rfloor < m-d$. Then, from $\per(S[1\dd m])=\per(S[d+1\dd d+m])= p< m- d$ we conclude $\per(S[1\dd d+m])=p$, and $U=S[1\dd m]=S[d+1\dd d+m]=V$.
\end{proof}

We define another auxiliary mapping $\rho_{\ell}$ for each layer $2\le \ell \le L$, which will be useful for dealing with highly periodic substrings.
\begin{definition}[Mapping $\rho$]
    \label{defn:mapG}
    For $2\le \ell\le L$ and $S \in \Sigma^{\tau_\ell}$, let
    \[ \rho_\ell(S):= \max \{r \in [1\dd  \tau_\ell] : \per(S[1\dd r]) = \per(S[1\dd \tau_{\ell-1}])\},\] 
    namely the furthest position that the period of the length-$\tau_{\ell-1}$ prefix can extend to.
\end{definition}
 We use $Y_0$ to denote $\{0,1\}^{\polylog(n)}$, which can hold the return value of $\pi_\sigma$ (a tuple) and the return value of $\rho_{\ell}$ (an integer). From now on we regard  $\pi_\sigma$ and $\rho_{\ell}$ as mappings from strings to $Y_0$.

Now we define the mapping $\psi$ by concatenating the values returned by $\pi$ and $\rho$ in all layers.

\begin{definition}[Mapping $\psi$]
    \label{defn:psi}
Define mapping $\psi_{\seed}\colon \Sigma^\tau \to Y_0^{2L-1}$ as follows: given $S[1\dd \tau]$, let 
\[\psi_{\seed}(S):= p_1 r_2  p_2  r_3 p_3 \dots r_{L}p_L,\]
where 
\[ p_\ell = \pi_{\seed}(S[1\dd \tau_\ell]), \, 1\le \ell \le L,\]
and 
\[ r_\ell = \rho_{\ell}(S[1\dd \tau_\ell]),\, 2\le \ell \le L.\]
\end{definition}

Finally, we define a mapping $\phi\colon \Sigma^\tau \to Y^{2L-1}$ by simply hashing each symbol in $\psi(S)$, where $Y = [2^{\polylog(n)}]$.
\begin{definition}[Mapping $\phi$]
    \label{defn:M}
Let seed $\seedDS \in \{0,1\}^w$. Let another seed $\seedH \in \{0,1\}^{(2L-1)\cdot O(\log n)}$ independently sample $2L-1$ min-wise independent hash functions $h_1,h_2,\dots,h _{2L-1} \colon Y_0 \to Y$ from \cref{lem:minwise} with parameters $n$ and $|Y_0|$. 

Then, define mapping $\phi_{\seedDS,\seedH}\colon \Sigma^\tau \to Y^{2L-1}$ as follows: given $S\in \Sigma^\tau$ with
\[ \psi_{\seedDS}(S) = y_1y_2\cdots y_{2L-1},\]
define
\[ \phi_{\seedDS,\seedH}(S) := h_1(y_1)h_2(y_2)\cdots h_{2L-1}(y_{2L-1}).\]

\end{definition}

We sometimes omit the random seeds $\sigma,\seedDS,\seedH$ from the subscripts and simply write $\psi(S),\phi(S)$.

We treat $\phi(S)\in Y^{2L-1}$ (and $\psi(S)\in Y_0^{2L-1}$) as length-$(2L-1)$ strings over alphabet $Y$ (and $Y_0$).
Elements in $Y^{2L-1}$ are compared by their lexicographical order.

For $1\le j\le 2L-1$, we also use $\phi(S)[1\dd j]\in Y^j$ to denote the length-$j$ prefix of $\phi(S)$, and use $\phi(S)[j] \in Y$ to denote the $j$-th symbol of $\phi(S)$. Notations $\psi(S)[1\dd j], \psi(S)[j]$ are defined similarly.

\subsection{Analysis of Sparsity}
In this section, we analyze the sparsity of the synchronizing set $\sA$ defined using mapping $\phi$ from \cref{defn:M}.

Recall that $\sQ$ (defined in \eqref{eqn:defnQ}) contains indices $i$ with $\per(T[i\dd i+\tau))\le \tau/3$, and recall that $\im(i):= \phi(T[i\dd i+\tau))$. We additionally define $\Psi(i):= \psi(T[i\dd i+\tau))$.

We have the following key lemma.

\begin{lem}[Expected count of minima]
    \label{lem:window-count}
    Suppose interval $[l \dd r] \subseteq [1\dd n-\tau+1]$ has length $r-l+1\le  \tau/4 $, and $[l\dd r] \cap \sQ = \emptyset$.

    Then, 
   \[ \Ex_{\seed_{\mathsf{DS}},\seed_{H}} \Big \lvert \Big \{ k \in [l\dd r] : \im(k) = \min \{\im(i): i\in [l\dd k] \}\Big \}\Big \rvert \le (O(\log \tau))^{2L-1},\]
   and 
   \[ \Ex_{\seed_{\mathsf{DS}},\seed_{H}} \Big \lvert \Big \{ k \in [l\dd r] : \im(k) = \min \{\im(i): i\in [k\dd r] \}\Big \}\Big \rvert \le (O(\log \tau))^{2L-1}.\]
\end{lem}
\begin{proof}
    We assume that the seed $\seed_{\mathsf{DS}}$  is successful for all substrings of $T$.  This happens with high probability over random $\seed_{\mathsf{DS}}\in \{0,1\}^w$ with length $w=O(\log n)$.

    First, we show that for $k_1,k_2 \in [l \dd r], k_1\neq k_2$ we must have $\Psi(k_1)\neq \Psi(k_2)$. 
   Suppose for contradiction that there are $k_1,k_2 \in [l \dd r], k_1<k_2$ such that $\psi_{\seed_{\mathsf{DS}}}(T[k_1\dd k_1+\tau)) = \psi_{\seed_{\mathsf{DS}}}(T[k_2\dd k_2+\tau)) $. Then, by definition of $\psi$ (\cref{defn:psi}), we have $\pi_{\seed_{\mathsf{DS}}}(T[k_1\dd k_1+\tau)) = \pi_{\seed_{\mathsf{DS}}}(T[k_2\dd k_2+\tau))$. 
    By \cref{lem:vishkin}, since $k_2-k_1\le \tau/4$, we must have $T[k_1\dd k_1+\tau) = T[k_2\dd k_2+\tau)$, and $(k_2-k_1)$ is a period of $T[k_1\dd k_1+\tau)$. This implies $k_1 \in \sQ$, contradicting  $[l\dd r] \cap \sQ = \emptyset$.

  In the following, we will only prove the first expectation upper bound. The second statement can be proved similarly.

    Recall that $\Phi(k)\in Y^{2L-1}$ is obtained by applying injective min-wise independent hash functions to $\Psi(k)\in Y_0^{2L-1}$, and they are compared in lexicographical order.
 In order for $\im(k)$ to be the lexicographically minimum among $\{\im(i)\}_{i \in [l\dd k]}$, its first symbol $\im(k)[1]$ needs to be the minimum among $\{\im(i)[1]\}_{i \in [l\dd k]}$, and then to break ties we compare the second symbols, and so on.
  Since the $j$-th symbol $\im(k)[j]$ is defined as the hash value $h_{j}(\Psi(k)[j])$, the probability of $\im(k)[j]$ being minimum at the $j$-th level is inversely proportional to the count of distinct $\Psi(i)[j]$ values that are being compared with.

Formally, for every $k\in [l\dd r]$ and every $1\le j\le 2L-1$, define the count
\[ c_{j,k} := \big \lvert  \{\Psi(k')[j] : k' \in [l \dd k], \Psi(k')[1\dd j-1]=\Psi(k)[1\dd j-1] \}\big \rvert.\]
Then, conditioned on $\Phi(k)[1\dd j-1] = \min \{\Phi(k')[1\dd j-1]: k' \in [l\dd k]\}$, we have $\Phi(k)[1\dd j] = \min \{\Phi(k')[1\dd j]: k' \in [l\dd k]\}$ if and only if $\Phi(k)[j]=h_j(\Psi(k)[j])$ becomes the minimum among the $c_{j,k}$ candidates currently in a tie, which happens with at most $1.1/c_{j,k}$ probability by the  min-wise independence of $h_j$. Hence, the probability 
that \[\im(k) = \min \{\im(i): i\in [l\dd k] \}\] happens is at most %
\begin{equation}
    \frac{(1 + 0.1)^{2L-1}}{c_{1,k}\cdot c_{2,k}\dots c_{2L-1,k}}.
    \label{eqn:proba}
\end{equation}
We will derive the desired expectation upper bound by summing \eqref{eqn:proba}  over all $k\in [l\dd r]$.

To do this, consider the following process of inserting strings $\Psi(l),\Psi(l+1),\dots,\Psi(r) \in Y_0^{2L-1}$ one by one into a \emph{trie}. Initially the trie only has a root node with node weight 1. Let $\Psi(k) \in Y_0^{2L-1}$ be the current string to be inserted. Starting from the root node, we iterate over $j=1,2,\dots,2L-1$, and each time move from the current node $p$ down to its child node $c$ labeled with symbol $\Psi(k)[j]\in Y_0$. If $p$ does not have such a child $c$, then we have to first create a child node $c$ labeled with $\Psi(k)[j]$, and assign $c$ a real weight $1/\deg(p)$, where $\deg(p)$ denotes the number of children (including $c$) that $p$ currently has.  
Now, observe that the value $1/\big (c_{1,k}c_{2,k}\dots c_{2L-1,k}\big )$ is upper-bounded by the product of node weights on the path from the root node to the leaf representing $\Psi(k)$. This is because $c_{j,k}$ equals the number of children that the $j$-th node on this path currently has, which is not smaller than the reciprocals of the weights of its children (in particular, the $(j+1)$-st node on this path).

Each node in the trie has at most $r-l+1\le \lfloor \tau/4\rfloor $ children, so the sum of the weights of its children is at most $1+1/2+1/3+\dots+1/(\lfloor \tau/4\rfloor) < \log \tau$. Since the trie has depth $2L-1$, a simple induction with distributive law of multiplication shows that  the sum of node weight products over all root-to-leaf paths in the trie is at most $(\log \tau)^{2L-1}$.

We previously showed that $\Psi(k)$ are distinct for all $k\in [l\dd r]$, so each leaf in the trie corresponds to only one $k\in [l\dd r]$. Hence, the sum of $1/\big (c_{1,k}c_{2,k}\dots c_{2L-1,k}\big )$ over all $k\in [l\dd r]$ is upper-bounded by $(\log \tau)^{2L-1}$. Then, the proof follows from summing over \eqref{eqn:proba}, which sums to at most $(1.1)^{2L-1} \cdot (\log \tau)^{2L-1} \le (O(\log \tau))^{2L-1}$.
\end{proof}

Now we are ready to prove the sparsity of the synchronizing set $\sA$ defined in \eqref{eqn:defnS}.
\begin{lem}[Sparsity]
    For every $i\in [1\dd n-3\tau+2]$, we have
   \[ \Ex_{\sigma_{\mathsf{DS}},\sigma_H} \big \lvert \sA\cap [i\dd i+\tau) \big \rvert \le (O(\log \tau))^{2L-1}.\]
    \label{lem:spar}
\end{lem}
\begin{proof}
    Recall that for $k\in [i\dd i+\tau)$, we have $k\in \sA$ if and only if 
    \[ \min \{\im(j): j\in [k\dd k+\tau]\setminus \sQ \} \in \{\im(k),\im(k+\tau)\}.  \]
To bound the number of such $k$,   we will separately bound and sum up the number of $k\in [i\dd i+\tau)$ such that 
   \begin{equation}
     \min \{\im(j): j\in [k\dd k+\tau]\setminus \sQ \} = \im(k),  
     \label{eqn:tobound}
   \end{equation}
    and the number of $k'\in [i+\tau\dd i+2\tau)$ (where we replaced $k=k'-\tau$) such that 
    \[ \min \{\im(j): j\in [k'-\tau\dd k']\setminus \sQ \} = \im(k').  \]
    In the following, we will only bound the first case (the number of $k\in [i\dd i+\tau)$ satisfying \eqref{eqn:tobound}). The second case is symmetric and can be bounded similarly.

    First we have the following claim.
    \begin{claim}
        \label{myclaim}
    For all $k\in [i\dd i+\tau)$ satisfying \eqref{eqn:tobound}, it holds that $k\notin \sQ$. 
    \end{claim}
    \begin{proof} If $k\in \sQ$, then there must exist $j\notin \sQ$ such that $\im(j)=\im(k)$, which implies $\pi_{\seed_{\mathsf{DS}}}(T[k\dd k+\tau))=\pi_{\seed_{\mathsf{DS}}}(T[j\dd j+\tau))$ by injectivity of the hash function $h_{L}$, and hence $\per(T[k\dd k+\tau)) = \per(T[j\dd j+\tau)) >\tau/3$, contradicting $k\in \sQ$. 
    \end{proof} 

    Hence, we only need to bound the number of $k\in [i\dd i+\tau) \setminus \sQ$, which by \cref{lem:structureQ} can be decomposed as the disjoint union of $O(1)$ many intervals, each having length at most $\tau/4$.
   For each such interval $[l'\dd r']$ (where $r'-l'+1\le \tau/4$ and $[l'\dd r'] \cap \sQ=\emptyset$), any $k\in [l'\dd r']$ satisfying \eqref{eqn:tobound} must also satisfy 
    $ \min \{\im(j): j\in [k\dd r']\} = \im(k)$. By \cref{lem:window-count}, the expected number of such $k$ in $[l'\dd r']$ is at most $(O(\log \tau))^{2L-1}$.  Summing over $O(1)$ intervals $[l'\dd r']$, the total expected number of such $k$ is also at most $(O(\log \tau))^{2L-1}$.
\end{proof}

\subsection{Algorithm for Reporting Synchronizing Positions}
In this section we will give an efficient quantum  algorithm for reporting the synchronizing positions in $\sA$.

First, observe that the mappings $\pi$ (\cref{defn:mapF}) and $\rho$ (\cref{defn:mapG}) can be computed efficiently.
\begin{obs}
    \label{obs:algo}
For $S\in \Sigma^m$ and seed $\sigma \in \{0,1\}^{O(\log n)}$, $\pi_{\seed}(S)$ can be computed in $\tilde O(\sqrt{m})$ quantum time.

For $2\le \ell \le L$ and $S\in \Sigma^{\tau_\ell}$, $\rho_{\ell}(S)$ can be computed in $\tilde O(\sqrt{\tau_\ell})$ quantum time.

Hence, for all $1\le \ell\le L, S\in \Sigma^{\tau_{\ell}}$, $\phi(S)[1\dd 2\ell-1]$ can be computed in $\tilde O(\sqrt{\tau_\ell})$ quantum time.
\end{obs}
\begin{proof}
The first statement immediately follows from \cref{lem:ds} and \cref{cor:per}. The second statement follows from \cref{cor:per} and a simple binary search with Grover search.
The third statement follows by direct computing.
\end{proof}

We will need to efficiently identify elements in $\sQ$. The following lemma will be useful.
\begin{lem}[Finding a Long Cubic Run]
    \label{lem:findrun}
   Given $S\in \Sigma^{n}$ with length $m\le n \le 4m/3$, there is a quantum algorithm in $\tilde O(\sqrt{n})$ time that finds (if exists) a \emph{maximal} substring  $S[i\dd j]$ that has length $j-i+1\ge m$ and period $\per(S[i\dd j]) \le m/3$, and such $S[i\dd j]$ is unique if exists.
\end{lem}
\begin{proof}
    Since $j-i+1\ge m$, $S[i\dd j]$ must contain $R:=S[n-m+1\dd m]$ as a substring. Since $|R| =2m-n \ge 2m/3 \ge 2\cdot \per(S[i\dd j])$, we must have $\per(S[i\dd j]) = \per(R)$. Hence, we can use \cref{cor:per} to compute $p = \per(R)$, and then uniquely determine $i,j$ by extending the period to the left and to the right, namely,
    \begin{align*}
     i &= \min \{i \in [1\dd n-m+1]: S[i']=S[i'+p] \text{ for all } i'\in [i\dd n-m+1]\},\\
     j &= \max \{j \in [m\dd n]: S[j']=S[j'-p] \text{ for all } j'\in [m\dd j]\},
    \end{align*}
    using binary search with Grover search. 
\end{proof}

\begin{lem}[Finding elements in $\sQ$]
    \label{lem:findq}
    Given $i\in [1\dd n-3\tau+2]$, one can compute $[i\dd i+2\tau)\cap \sQ$, represented as disjoint union of $O(1)$ intervals, in $\tilde O(\sqrt{\tau})$ quantum time.
\end{lem}
\begin{proof}
By \cref{lem:structureQ}, the set $[i\dd i+2\tau)\cap \sQ$ can be decomposed into disjoint union of $O(1)$ many intervals. To find them, we separately consider overlapping blocks of length $\lfloor 4\tau/3\rfloor$ inside $[i\dd i+2\tau)$, with every adjacent two blocks being at a distance of $\lfloor \tau/6 \rfloor$.
    There are $O(1)$ many such blocks, and for every $j\in [i\dd i+2\tau)\cap \sQ$ there is some block that contains $[j\dd j+\tau)$. Hence, we can apply \cref{lem:findrun} to each block, and find all $j$'s (represented as an interval) such that $T[j\dd j+\tau)$ is contained in this block and has period at most $\tau/3$.
\end{proof}

The main technical component is an algorithm that efficiently finds the minimal $\im(k)$ over an interval, stated in the following lemma.

\begin{lem}[Finding range minimum $\im(k)$]
    \label{lem:rangemin}
Given $[l \dd r] \subseteq [1\dd n-\tau+1]$ of length $r-l+1\le \tau/4$, we can compute 
\[\arg \min_{k\in [l\dd r]}\im(k),\] 
represented as an arithmetic progression, in 
\begin{equation}
    \label{eqn:final-bound}
 \sqrt{\tau} \cdot (O(\log \tau))^{L} \cdot \polylog(n)\cdot  \sum_{ \ell=1}^{L-1} \sqrt{\tau_{\ell+1}/\tau_{\ell}}
\end{equation}
 quantum time.
\end{lem}

Before proving \cref{lem:rangemin}, we show that it can be used to report all the synchronizing positions in a length-$\tau$ interval, one at a time.

\begin{lem}[Efficient reporting]
    \label{lem:report}
    Given $i\in [1\dd n-3\tau+2]$, the elements in $\sA\cap [i\dd i+\tau)$ can be reported in $O((\mathsf{cnt}+1)\cdot \cT)$ time, where $\mathsf{cnt}=|\sA\cap [i\dd i+\tau)|$ is the output count, and $\cT$ is the time bound \eqref{eqn:final-bound} in \cref{lem:rangemin}.
\end{lem}
\begin{proof}
To report $\sA\cap [i\dd i+\tau)$, it suffices to report all $k\in [i\dd i+\tau)$ such that 
   \begin{equation}
     \min \{\im(j): j\in [k\dd k+\tau]\setminus \sQ \} = \im(k),  
     \label{eqn:toreport}
   \end{equation}
    and all $k'\in [i+\tau\dd i+2\tau)$ (where we replaced $k=k'-\tau$) such that 
    \[ \min \{\im(j): j\in [k'-\tau\dd k']\setminus \sQ \} = \im(k').  \]
We will only show the first part. The second part is symmetric and can be solved similarly.

By \cref{lem:findq}, the set $[i\dd i+2\tau)\setminus \sQ$ can be decomposed into disjoint union of $O(1)$ many intervals, which can be found in $\tilde O(\sqrt{\tau})$ quantum time.
Then, we can  use \cref{lem:rangemin} to answer $\arg \min \{ \im(j): j\in [l\dd r] \setminus \sQ\}$ given any $[l\dd r] \subseteq [i\dd i+2\tau)$, in $O(\mathcal{T})$ time.

By \cref{myclaim} (note that \eqref{eqn:toreport} is identical to \eqref{eqn:tobound}), for all $k\in [i\dd i+\tau)$ satisfying \eqref{eqn:toreport} it holds that $k \notin \sQ$.
To report all $k\in [i\dd i+\tau)
 $ 
  satisfying \eqref{eqn:toreport}, we follow the pseudocode in Algorithm~\ref{alg:report}.

    \begin{algorithm2e}
\caption{Reporting  all $k\in [i\dd i+\tau)$ satisfying \eqref{eqn:toreport}}
\label{alg:report}
\Begin{
Let $l \gets i$\;
\While{$l<i+\tau$}{
        Let $k$  be the minimum element in  $\arg \min \{ \im(j): j\in [l\dd i+\tau) \setminus \sQ\}$\;
        \If {$k$ satisfies \eqref{eqn:toreport}}{
            report $k$\;
            $l \gets k+1$\;
        }\Else{
            \textbf{break}
        }
}
}
\end{algorithm2e} 
It is clear that the running time of Algorithm~\ref{alg:report} is $O((\mathsf{cnt}+1)\cdot \cT)$, where $\mathsf{cnt}$ is the number of reported elements.
It suffices to show that Algorithm~\ref{alg:report} does not miss any $k$ satisfying \eqref{eqn:toreport}. We maintain the invariant that all not yet reported $k$ satisfying \eqref{eqn:toreport} are at least $l$. Whenever we report $k$, it holds that $\im(k')>\im(k)$ for all $k'\in [l\dd k)$, so none of such $k'$ can satisfy \eqref{eqn:toreport}, and it is safe to skip them by letting $l\gets k+1$. When the line of \textbf{break} is reached, we have $k\in \arg \min \{ \im(j): j\in [l\dd i+\tau) \setminus \sQ\}$  but $k\notin  \arg \min \{\im(j): j\in [k\dd k+\tau) \setminus \sQ\}$, implying that there exists $r \in [i+\tau\dd k+\tau)$ such that $\im(r)<\im(k')$ for all $k'\in [k\dd i+\tau)$, so none of $k'\in [k\dd i+\tau)$ can satisfy \eqref{eqn:toreport}, and we can terminate the procedure.
\end{proof}

Now it remains to describe the algorithm in \cref{lem:rangemin} for finding the minimum $\im(k)$ in an interval. Our algorithm is recursive, stated as follows.
\begin{lem}
    \label{lem:rangemin-recurs}
    For $1\le \ell\le L$, there is a quantum algorithm running in $\cT_\ell$ time that, given $[l \dd r] \subseteq [1\dd n-\tau_\ell+1]$ of length $r-l+1\le \tau_\ell/4$, computes 
\[\arg \min_{k\in [l\dd r]}\im(k)[1\dd 2\ell-1]\] 
represented as an arithmetic progression. The time complexities satisfy
\begin{equation}
 \cT_\ell \le \sqrt{\tau_\ell/\tau_{\ell-1}}\cdot \big (\cT_{\ell-1}\cdot O(\log \tau)+\sqrt{\tau_\ell}\cdot  \polylog(n) \big )
 \label{eqn:bound}
\end{equation}
for $\ell >1$, and $\cT_1=\polylog(n)$.
\end{lem}

\begin{proof}
    We first explain why the minimizers $k$ form an arithmetic progression. For any $l\le k_1<k_2\le r$ with $\Phi(k_1)[1\dd 2\ell-1]=\Phi(k_2)[1\dd 2\ell-1]$, we have $\pi_{\sigma_{\mathsf{DS}}} (T[k_1\dd k_1+\tau_\ell)) = \pi_{\sigma_{\mathsf{DS}}} (T[k_2\dd k_2+\tau_\ell))$, and from \cref{lem:vishkin} we have $T[k_1\dd k_1+\tau_\ell) =T[k_2\dd k_2+\tau_\ell)$ since $k_2-k_1\le \tau_\ell/4$. Hence, the minimizers $k$ in this length-$\tau_\ell/4$ interval correspond to the starting positions  of all occurrences of a certain length-$\tau_\ell$ substring, which form an arithmetic progression by \cref{lm:period-patterns}.

Due to lexicographical comparison rule, we know that any element in $\arg\min_{k\in [l\dd r]} \Phi(k) [1\dd 2\ell-1]$ must also be in $\arg\min_{k\in [l\dd r]} \Phi(k) [1\dd 2(\ell-1)-1]$, so we can recursively use the $(\ell-1)$-st level subroutine to find candidate minimizers. We first consider the simpler  case where each $(\ell-1)$-st level subroutine only returns a single minimizer.
  Given the input interval $[l\dd r]$ of length $r-l+1\le \tau_\ell/4$, we divide $[l\dd r]$ into blocks each of length $\tau_{\ell-1}/4$, on which we recursively apply the $(\ell-1)$-st level subroutine (boosted to $1-1/\poly(\tau)$ success probability by $O(\log \tau)$ repetitions) to find the minimizer $k'$ in this block, and then use $\sqrt{\tau_\ell}\cdot \polylog(n)$ time to compute $\Phi(k')[1\dd 2\ell-1]$.
  Out of the $\tau_\ell/\tau_{\ell-1}$ blocks, we apply quantum minimum finding to find the block that has the minimum $\Phi(k')[1\dd 2\ell-1]$ (\cref{obs:algo}). This quantum minimum finding procedure incurs $O(\sqrt{\tau_\ell/\tau_{\ell-1}})$ comparisons, and results in the time bound \eqref{eqn:bound}.

  Now we consider the general case, where for a $(\ell-1)$-st level block $[l'\dd r']$, $\arg\min_{k\in [l'\dd r']} \Phi(k) $ $[1\dd 2(\ell-1)-1]$ may contain multiple elements $k_1<k_2<\dots<k_d$ forming an arithmetic progression, where $k_d-k_1\le \tau_{\ell-1}/4$.
  We already know from \cref{lm:period-patterns} that $T[k_j\dd k_j+\tau_{\ell-1})$ are identical for all $j\in [d]$,  and $p=k_2-k_1 = \dots = k_d - k_{d-1} = \per(T[k_1\dd k_d+\tau_{\ell-1}))$.
   To find the minimum $\Phi(k)[1\dd 2\ell-1]$ among $k\in \{k_1,\dots,k_d\}$, we will additionally compare $\Phi(k)[2(\ell-1)] = h_{2(\ell-1)}\big (\rho_{\ell}(T[k\dd k+\tau_{\ell}))\big)$ to break ties.

   Now, let $r' = \max \{r' : r'\le k_d+\tau_{\ell}, \per(T[k_1\dd r') ) = p \}$, namely the furthest point this period can extend to (up to a distance of the current level interval length $\tau_{\ell}$). 
Now we can easily obtain their $\rho_{\ell}$ values,
       \[
           \rho_{\ell} (T[k_j\dd k_j+\tau_{\ell}) ) =  \min \{ \tau_{\ell}, r' - k_j \}.
           \]
          Then, among strings $T[k_j\dd k_j+\tau_{\ell})$, $j\in [d]$, those with $\rho_{\ell}$ value equal to $\tau_{\ell}$ are all identical to $T[k_1\dd k_1+\tau_{\ell})$, and the remaining ones must have different $\rho_{\ell}$ values that are strictly smaller than $\tau_{\ell}$. Hence, we can simply use Grover search over $j\in [d]$ in $\tilde O(\sqrt{d})\le \tilde O(\sqrt{\tau_{\ell-1}})$ time to find the minimum
          $\Phi(k_j)[2(\ell-1)] = h_{2(\ell-1)}\big (\rho_{\ell}(T[k_j\dd k_j+\tau_{\ell}))\big)$.
\end{proof}

\begin{proof}[Proof of \cref{lem:rangemin} assuming \cref{lem:rangemin-recurs}]
   Note that the algorithm in \cref{lem:rangemin-recurs} with $\ell=L$ solves the task required by \cref{lem:rangemin}. The quantum time complexity $\cT_L$ can be calculated by expanding \eqref{eqn:bound} into 
   \begin{align*}
    \cT_L  & \le \sqrt{\tau_L} \sqrt{\frac{\tau_L}{\tau_{L-1}}}\polylog(n) + \sqrt{\frac{\tau_L}{\tau_{L-1}}} \cT_{L-1}\cdot O(\log \tau)\\
    & \le \sqrt{\tau_L} \sqrt{\frac{\tau_L}{\tau_{L-1}}}\polylog(n) + \sqrt{\tau_L}\sqrt{\frac{\tau_{L-1}}{\tau_{L-2}}}\polylog(n)\cdot O(\log \tau) +  \sqrt{\frac{\tau_L}{\tau_{L-2}}}\cT_{L-2}(O(\log \tau))^2\\
    & \le \cdots\\
    & \le \sqrt{\tau_L} \sum_{\ell=1}^{L-1} \sqrt{\frac{\tau_{\ell+1}}{\tau_{\ell}}}\polylog(n)\cdot (O(\log \tau))^{L-\ell-1}. 
   \qedhere 
   \end{align*}
\end{proof}

Finally, we prove the main theorem \cref{thm:quantum-sync} by setting the length parameters $\tau_1,\dots,\tau_L$ appropriately.

\begin{proof}[Proof of \cref{thm:quantum-sync}]
     We set $L = \Theta(\sqrt{\log \tau})$, and pick $6= \tau_1<\tau_2<\dots<\tau_L=\tau$ so that $\tau_{\ell+1}/\tau_{\ell} \le O(\tau^{1/L})$ for all $1\le \ell <L$. Then, by \eqref{eqn:final-bound}, the quantum time complexity of the range minimum subroutine (\cref{lem:rangemin}) is at most
    \begin{align*}
         & \sqrt{\tau} \cdot (O(\log \tau))^{L} \cdot \polylog(n)\cdot L\cdot  O(\tau^{1/(2L)}) \\
         \le \ &  \sqrt{\tau}\cdot \polylog(n) \cdot  (O(\log \tau))^{O(\sqrt{\log \tau})} \\
         \le \ & \tau^{1/2+o(1)}\cdot \polylog(n). 
        \label{eqn:actual-bound}
    \end{align*}
    Recall that with high probability $\sigma_{\mathsf{DS}}\in \{0,1\}^{O(\log n)}$ succeeds for all substrings $S$ of $T$.
        The rest of the proof follows from \cref{lem:kk}, \cref{lem:spar}, and \cref{lem:report}.
\end{proof}

\section{Longest Common Substring}

\label{sec:lcs}
In the Longest Common Substring problem with threshold $d$, we are given two strings $S_1,S_2 \in \Sigma^n$, and need to decide whether $S_1$ and $S_2$ have a common substring $S_1[i\dd i+d) = S_2[j\dd j+d)$ for some $i \in [n-d+1], j \in [n-d+1]$.

In \cref{sec:lcsub} we first show that, by assembling our quantum string synchronizing set with known techniques, we can obtain improved quantum time complexity for LCS with threshold $d$.
Then, in \cref{sec:lcslb} we apply a known result in quantum query complexity to show a lower bound for LCS with threshold $d$ that matches the upper bound up to $n^{o(1)}$ factors.

\subsection{Improved Upper Bound}
\label{sec:lcsub}

Given input strings $S_1,S_2$, define $T=S_1\$S_2$ as their concatenation separated by a delimiter symbol, and let $|T|=N$.
 Assume the threshold length $d$ satisfies $d\ge 100$ (otherwise one can simply apply the algorithm in \cite{legall}). 

Akmal and Jin \cite{aj22} solved the LCS with threshold $d$ problem using the anchoring technique 
\cite{DBLP:conf/cpm/StarikovskayaV13,DBLP:conf/cpm/Charalampopoulos18,DBLP:conf/esa/AmirCPR19,DBLP:journals/algorithmica/AmirCPR20,DBLP:conf/cpm/Nun0KK20,icalp20,lcs2021,aj22}. In this technique, one constructs a set of anchors, and only focuses on finding anchored common substrings, defined as follows.

\begin{definition}[Anchor sets]
	\label[definition]{defn:good-anchor}
For  $T = S_1\$ S_2$ of length $|T|=N$ and  subset $\sC \subseteq [1\dd N]$, a common substring $S_1[i_1\dd i_1+d)=S_2[i_2\dd i_2+d)$ is said to be \emph{anchored} by $\sC$, if there exists a shift $h\in [0\dd d)$ such that $i_1+h,  |S_1|+1+i_2+h \in \sC$. 

Given $T = S_1\$ S_2$ and threshold length $d$, we say $\sC\subseteq [1\dd N]$ is an \emph{anchor set} if the following holds: if $\mathrm{LCS}(S_1,S_2) \ge d$, then $S_1$ and $S_2$ must have a length-$d$ common substring anchored by $\sC$.
\end{definition}

The following is the main technical result of Akmal and Jin \cite{aj22}, obtained using quantum walk search on Johnson graphs \cite{DBLP:journals/siamcomp/Ambainis07,mnrs}. It shows that a strongly explicitly computable anchor set of small size can be used to solve the LCS with threshold $d$ problem.
\begin{them}[\cite{aj22}]
	\label{thm:aj}
	Given input strings $S_1,S_2$ and threshold $d$, suppose there is a
	$\mathcal{T}$-time quantum algorithm that outputs $C(j) \in [N]$ given any $1\le j\le m$, such that  $\{C(1),C(2),\dots,C(m)\}$  is an anchor set.

Then, one can decide whether $\mathrm{LCS}(S_1,S_2)\ge  d$, in 
	\[ \tilde O(m^{2/3} \cdot (\sqrt{d} + \mathcal{T}))\]
    quantum time.
\end{them}

Akmal and Jin designed an anchor set with $\tilde O(\sqrt{d})$ reporting time but  suboptimal size, based on the classical string synchronizing set of Kempa and Kociumaka \cite{kk19}. Here we will use our quantum string synchronizing set to  improve their construction. To do this, we will reuse a key lemma from \cite{lcs2021}.

For threshold length $d$, we set the parameter $\tau = \lfloor d/3\rfloor $.
Following \cite{lcs2021}, define a \emph{$\tau$-run} is a run (see \cref{sec:prelim1}) of length at least $3\tau-1$ with period at most $\frac{1}{3} \tau$. Consider the following set $\sB$ of positions in $T=S_1\$S_2$.
\begin{definition}[\cite{lcs2021}]
Define $\sB \subseteq [N]$ as follows: for all $\tau$-runs $T[l\dd r]$, let $P$ be the Lyndon root of $T[l\dd r]$, and $P=T[i^{(b)}\dd i^{(b)}+p)=T[i^{(e)}\dd i^{(e)}+p)$ be the first and last occurrences of $P$ in $T[l \dd r]$. Add $i^{(b)},i^{(b)}+p$, and $i^{(e)}$ into $\sB$.
\end{definition}

Then, \cite{lcs2021} showed the following lemma.
\begin{lem}[{\cite[Lemma 15]{lcs2021}}]
	\label{lem:esa}
Let $\sA$ be a $\tau$-string synchronizing set of $T =S_1\$ S_2$. Then, $\sA \cup \sB$ is an anchor set.

More specifically, let $S_1[i_1\dd i_1+d)=S_2[i_2\dd i_2+d)$ be a length-$d$ common substring with minimized $i_1+i_2$. Then, $S_1[i_1\dd i_1+d)=S_2[i_2\dd i_2+d)$ is anchored by $\sA \cup \sB$.
\end{lem}

Hence, the rest of the task is to implement \cref{lem:esa} in a strongly explicit way, and apply \cref{thm:aj}. 
To compute Lyndon roots, we use recent quantum algorithms for minimal string rotation \cite{ying,aj22} (see also \cite{qishengwang,quantdc}).
 \begin{lem}[Minimal String Rotation \cite{aj22}]
	\label{lem:minrot}
	Given $S\in \Sigma^m$, computing the lexicographical minimal rotation of $S$ can be solved  in $m^{\frac{1}{2}+o(1)}$ quantum time.
 \end{lem}

\begin{definition}[Anchor set via quantum synchronizing set]
 Let $\sA$ be the $\tau$-synchronizing set of $T$ determined by random seed $\sigma$ (\cref{thm:quantum-sync}).
Let $f(\tau) = \tau^{o(1)}$ be the expected sparsity upper bound in \cref{thm:quantum-sync}.

Define $\sJ = \{1,1+\tau,1+2\tau,\dots \} \cap  [N-3\tau+2]$ of size $|\sJ| \le O(N/\tau)$.
For every $i\in  \sJ $, let $\sC_i\subseteq [1\dd N]$ be defined by the following procedure.
\label[definition]{defn:anchor}
\begin{itemize}
	\item \textbf{Step 1:} If  $|\sA \cap [i\dd i+\tau)| \le 100 f(\tau) $, then add all the elements from $A \cap [i\dd i+\tau)$ into $\sC_i$. Otherwise, do not add any.
	\item \textbf{Step 2:} 
	If $i\in  \sQ$ (defined in \eqref{eqn:defnQ}), then let $p:=\per(T[i\dd i+\tau)) \le \tau/3$, and extend this period to both directions (up to distance $\tau$): 
	    \begin{align*}
	        r &:= \max\big \{r: r\le \min \{N,i+2\tau\} \land \per(T[i\dd r])=p\big \},\\ l &:= \min\big \{l: l\ge \min \{1,i-\tau\} \land \per(T[l\dd i+\tau)) =p\big \}.
	    \end{align*} 
		Let $P$ be the Lyndon root of
	    $T[l\dd r]$. Let $P=T[i^{(b)}\dd i^{(b)}+p)=T[i^{(e)}\dd i^{(e)}+p)$ be the first and last occurrences of $P$ in $T[l \dd r]$. Add $i^{(b)},i^{(b)}+p$, and $i^{(e)}$ into $\sC_i$.
\end{itemize}
Finally, the anchor set $\sC$ is defined as $\bigcup_{i \in \sJ } \sC_i$.
\end{definition} 

Now we will show that the anchor set in \cref{defn:anchor} can be used in \cref{thm:aj} and lead to an improved algorithm for LCS with threshold $d$.

\lcsubrestate*

\begin{proof}[Proof of \cref{thm:lcsub}]
	Recall $N= |T|=2n+1$ and $\tau = \lfloor d/3\rfloor $.
We first observe that $|\sC_i|\le \tau^{o(1)}$ and $\sC = \bigcup_{i\in \sJ} \sC_i$ has size $|\sC| \le |\sJ|\cdot \tau^{o(1)} \le n/\tau^{1-o(1)}$ by definition. Then, observe that $\sC_i$ can be computed in $\tau^{1/2+o(1)}\cdot \polylog(n)$ quantum time given any $i\in \sJ$, due to the efficient computability of $\sA$ (\cref{thm:quantum-sync}), and fast quantum algorithms for computing period (\cref{cor:per}) and minimal string rotation (\cref{lem:minrot}). Hence, the time bound  in \cref{thm:aj} is 
\[ \tilde O(|\sC|^{2/3} \cdot (\sqrt{d}+\mathcal{T})) \le \tilde O((n/\tau^{1-o(1)})^{2/3} \cdot (\sqrt{d}+\tau^{1/2+o(1)})) \le \tilde O(n^{2/3}/d^{1/6-o(1)}).\]

It remains to explain why $\sC$ is an anchor set with at least constant probability. 

First we show that $\sB \subseteq \sC$. 
Let any $\tau$-run $T[l\dd r]$ be given, with
$P=T[i^{(b)}\dd i^{(b)}+p)=T[i^{(e)}\dd i^{(e)}+p)$ being the first and last occurrences of its Lyndon root $P$ in $T[l \dd r]$.
By definition of $\sQ$, we have $[l\dd r-\tau+1]\subseteq \sQ$. Let $j_1,j_2$ be the minimum and maximum  $j\in \sJ \cap [l\dd r-\tau+1]$ (which must be non-empty).
Then, we must have $j_1 - l\in [0\dd \tau)$ which then implies $i^{(b)}, i^{(b)}+p \in \sC_{j_1}$ in Step 2 of \cref{defn:anchor}. Similarly, we can show $i^{(e)} \in \sC_{j_2}$.
Hence, $\sB \subseteq \sC$.

Now, let $S_1[i_1\dd i_1+d)=S_2[i_2\dd i_2+d)$ be a length-$d$ common substring with minimized $i_1+i_2$. If $S_1[i_1\dd i_1+d)=S_2[i_2\dd i_2+d)$ is already anchored by $\sB$, then we are done. Otherwise, by \cref{lem:esa} it is anchored by $\sA \cup \sB$, for any $\tau$-synchronizing set $\sA$.
We consider the synchronizing positions  in $\sA$ that may be used as the anchor here, i.e., they are close to $i_1$ or $|S_1|+1+i_2$ up to $O(\tau)$ distance.
By the sparsity property (\cref{thm:quantum-sync}), the expectation (over seed $\sigma$) of 
\[|[i_1-\tau\dd i_1+d+\tau) \cap \sA| + |[|S_1|+1+i_2-\tau\dd |S_1|+1+i_2+d+\tau)  \cap \sA|  \]
is at most $10\cdot f(\tau)$. By Markov's inequality, with constant probability, this is at most $100 f(\tau)$, meaning that these anchors will all be included in $\sC$ by the Step 1 of \cref{defn:anchor}. Hence, $\sC$ is an anchor set with at least constant probability.
\end{proof}

\subsection{A Matching Lower Bound}
\label{sec:lcslb}

In this section, we observe a nearly  matching lower bound for the LCS problem with threshold $d$.

\lcslbrestate*

Our proof is based on known results in quantum query complexity. Let 
\[ P_d := \{ x_1x_2\cdots x_d \in (\Sigma\cup \{\star\})^d : \text{ exactly one $i\in [d]$ satisfies $x_i \neq \star$} \}.\]
Define a partial function $\mathsf{PSEARCH}_d \colon P_d \to \Sigma$ by letting $\mathsf{PSEARCH}_d(x_1\cdots x_d)$ return the only non-$\star$ symbol $x_i$ ($i\in [d]$).
We will use the following composition theorem with inner function being  the $\mathsf{PSEARCH}$ problem \cite{psearch}.
\begin{them}[\cite{psearch}]
	\label{thm:psearch}
    Let $f\colon \Sigma^m \to A$ be a function with quantum query complexity $Q(f)$. Let $h$ be the composition of $f$ with $\mathsf{PSEARCH}_d$, namely
    \[ h(x_{1,1},\dots,x_{1,d}; \, \dots \, ; x_{m,1},\dots,x_{m,d} ) = f(\mathsf{PSEARCH}_d(x_{1,1},\dots,x_{1,d}),\dots,\mathsf{PSEARCH}_d(x_{m,1},\dots,x_{m,d})).\]
    Then, the quantum query complexity of $h$ satisfies $Q(h) \ge \Omega( Q(f)\cdot \sqrt{d})$.
\end{them}

We will reduce the composition of (bipartite) Element Distinctness with  $\mathsf{PSEARCH}_d$ to the LCS problem with threshold $d$.

\begin{proof}[Proof of \cref{thm:lcslb}]
	In the (bipartite) Element Distinctness problem, we are given two length-$m$ arrays $a,b \in \Sigma^m$, and want to decide whether there exist $i,j\in [m]$ such that  $a_i=b_j$. When $|\Sigma| \ge \Omega(m)$, this problem requires $\Omega(m^{2/3})$ quantum query complexity \cite{DBLP:journals/jacm/AaronsonS04,DBLP:journals/toc/Ambainis05}.

    Let $m=n/d$. Let $a,b\in (\Sigma\cup \{\star\})^{md}$ be an instance of Bipartite Element Distinctness problem composed with the $\mathsf{PSEARCH}_d$ problem.
By \cref{thm:psearch},  this problem requires quantum query complexity $\Omega(m^{2/3}\cdot \sqrt{d}) = \Omega(n^{2/3}/d^{1/6})$.
    
    Now we will create an LCS instance of two strings $S,T$ of length $\Theta(n)$. Define string $S$ as the concatenation
	\[ S:=  A_1 \#_1 A_2 \#_2 \cdots \#_{m-1} A_m,\]
	where $\#_i$ are distinct delimiter symbols, and block $A_i$ $(i\in [m])$ is defined as
    \[ A_i = \star^{10d} a_{i,1}a_{i,2}\dots a_{i,d}  \star^{10d}.\]
    Notice that, when $a_{i,1}a_{i,2}\dots a_{i,d}$ is an input for the $\mathsf{PSEARCH}_d$ problem, the unique non-$\star$ character inside block $A_i$ has at least $10d$ $\star$s next to it on both sides.
    The string $T$ is similarly defined based on $b$, using a different set of delimiter symbols.

If there exits $i,j\in [m]$ such that $\tilde a_i = \tilde b_j \neq \star$, where $\tilde a_i=\mathsf{PSEARCH}_d(a_{i,1},\dots,a_{i,d})$, and $\tilde b_j=\mathsf{PSEARCH}_d(b_{j,1},\dots,b_{j,d})$, then the LCS between $S,T$ is at least $1+20d$, consisting of the matched non-$\star$ symbol together with $\star^{10d}$ on its left and another $\star^{10d}$ on its right.
 On the other hand, any substring in $S$ (or $T$) of length at least $11d$ that contains no delimiters must contain a non-$\star$ symbol, so a common substring between $S$ and $T$ of length at least $11d$ implies that $a$ and $b$ have a common non-$\star$ symbol. 

 Hence, we can solve the instance by checking whether $\mathrm{LCS}(S,T) \ge 20d+1$ or $\mathrm{LCS}(S,T) \le 11d-1$.
\end{proof}

Observe that the proof easily extends to $(2-\eps)$-approximation for any constant $\eps>0$.
By a suitable binary encoding of $\Sigma$, one can extend the lower bound to the case of binary strings with a $O(\log n)$-factor loss, similar to \cite{legall}.

\section{Data Structure for LCE Queries}	
\label{sec:lce}
In the LCE problem, one is given a string $T\in \Sigma^n$ and needs to preprocess a data structure $D$, so that later one can efficiently answer $\lce(i,j)$ given any $i,j\in [n]$, defined as $\lce_T(i,j):= \lcp(T[i\dd n], T[j\dd n])$.

Kempa and Kociumaka \cite{kk19} used $\tau$-synchronizing sets to design an optimal data structure for LCE queries in the small-alphabet setting where the input string $T$ is given as a packed representation, in which each word holds $\tau = O(\log_{\Sigma} n)$ characters in the word-RAM model.
In this section, we observe that, using our quantum string synchronizing set (\cref{thm:quantum-sync}), the LCE data structure of \cite{kk19} can be adapted into the quantum setting, stated as follows.

\lcerestate*

\begin{proof}[Proof Sketch]
Let $\sA$ be the $\tau$-synchronizing set of $T$ from \cref{thm:quantum-sync} with expected size $\Ex_{\sigma}|\sA| = (n/\tau)\cdot \tau^{o(1)}$. We can assume that the seed $\sigma$ is successful, and $|\sA| \le 10 \Ex_{\sigma}|\sA|$, which happens with constant probability by Markov's inequality. Then we use the reporting algorithm to report all elements in $\sA$, in $(n/\tau + |\sA|)\cdot \tilde O(\tau^{1/2+o(1)}) \le \tilde O(n/\tau^{1/2-o(1)})$ quantum time.

Having computed all elements in $\sA$, the algorithm from now on is almost the same as \cite[Section 5]{kk19}, and we shall not repeat the algorithm description here. At a high level, \cite{kk19}'s preprocessing algorithm consists of $O(|\sA|)$ many basic operations on $\tau$-length substrings, including computing longest common prefix (and lexicographical comparison), pattern matching, and computing period. Their algorithm for answering $\lce_T(i,j)$ given any $i,j\in [n]$ takes $O(1)$ such basic operations. In our case, each basic operation on $\tau$-length substrings can be done in $\tilde O(\sqrt{\tau})$ quantum time (\cref{obs:lcp}, \cref{thm:matching}, \cref{cor:per}). Hence, we have $\mathcal{T}_{\mathsf{prep}} \le \tilde O(n/\tau^{1/2-o(1)}) + |\sA|\cdot \tilde O(\sqrt{\tau}) \le \tilde O(n/\tau^{1/2-o(1)})$, and $\mathcal{T}_{\mathsf{ans}} \le \tilde O(\sqrt{\tau})$.  (In comparison, in \cite{kk19}'s  small-alphabet setting, each $\tau$-length substring is packed as a machine word, and each basic operation takes $O(1)$ time in word-RAM.)
\end{proof}

The LCE data structure from \cref{thm:lce} can be used to answer whether two substrings in $T$ have Hamming distance at most $k$, via the standard kangaroo jumping method: starting from the leftmost endpoints, each time use an LCE query to determine the next Hamming mismatch, until hitting the rightmost endpoints or finding more than $k$ mismatches.
This data structure will be useful in our $k$-mismatch matching algorithm in \cref{sec:kmis}.

\begin{cor}
\label{thm:candidatepos-sync}
Given string $T\in \Sigma^n$ and parameter $1\le \tau \le n/2$, after $\tilde O(n/\tau^{\frac{1}{2}-o(1)})$ time quantum preprocessing, we can obtain a data structure that supports answering whether two substrings $T[l_1\dd l_1+d),T[l_2\dd l_2+d)$ have Hamming distance at most $k$ (where $k\ge 1$), in $\tilde O(k\sqrt{\tau})$ quantum time complexity.
\end{cor}

We observe that our LCE data structure in \cref{thm:lce} cannot be significantly improved. The proof follows from the following result in quantum query complexity.

\begin{them}[\cite{polynomial,Paturi92}]
   \label{thm:lbham} 
For $1\le k\le n/2$, given $x\in \{0,1\}^n$, deciding whether its Hamming weight $\|x\|_1 = \sum_{i=1}^n x_i$ satisfies $\|x\|_1\le k$ requires quantum query complexity $\Omega(\sqrt{kn})$.
\end{them}

\optrestate*
\begin{proof}
	Given $x\in \{0,1\}^n$, we can use \cref{thm:candidatepos-sync} (which directly follows from the LCE data structure in \cref{thm:lce}) to check whether $\|x\|_1\le k$, using $\tilde O(\mathcal{T}_{\mathsf{prep}} + k\cdot \mathcal{T}_{\mathsf{ans}})$ quantum queries.
	By \cref{thm:lbham}, we must have 
\[ \sqrt{kn} \le \tilde O(\mathcal{T}_{\mathsf{prep}} + k\cdot \mathcal{T}_{\mathsf{ans}}).\]

\begin{itemize}
	\item 
If $\mathcal{T}_{\mathsf{prep}} \in [\tilde \Omega(\sqrt{n}),\tilde O(n)]$, then we set $k = (\mathcal{T}_{\mathsf{prep}}^2 / n)\cdot \polylog(n) \in [1,n/2]$, so that $\mathcal{T}_{\mathsf{prep}}\cdot \polylog(n) <0.5 \sqrt{kn}$ and hence 
\[ k\cdot \mathcal{T}_{\mathsf{ans}} \ge \sqrt{kn}/\polylog(n) - \mathcal{T}_{\mathsf{prep}} \ge 0.5 \sqrt{kn}/\polylog(n),\]
which implies
\[ \mathcal{T}_{\mathsf{ans}} \ge \tilde \Omega (\sqrt{n/k}) = \tilde \Omega(n/\mathcal{T}_{\mathsf{prep}}).\]
Hence the statement holds.
\item If $\mathcal{T}_{\mathsf{prep}} \le  \tilde O(\sqrt{n})$, then setting $k=1$ in the above argument shows $\mathcal{T}_{\mathsf{ans}} \ge \tilde \Omega(\sqrt{n})$, which also implies the statement.
\item If $\mathcal{T}_{\mathsf{prep}} \ge  \tilde \Omega(n)$ then the statement is immediately true. \qedhere
\end{itemize}
\end{proof}

\section{The \texorpdfstring{$k$}{k}-Mismatch Matching Problem}

\label{sec:kmis}
In this section we will show how to devise an algorithm for the $k$-mismatch matching problem running in $\tilde{O}(k\sqrt{n})$ quantum time. Our approach heavily relies on structural insights described in \cite{ckw20}, which we will adapt to the quantum setting in \cref{sub:decomp}. After having performed a structural analysis of a pattern $P \in \Sigma^m$, we will show in \cref{sub:matching} how to report efficiently the existence of a $k$-mismatch occurrence of $P$ in a text $T \in \Sigma^n$. Finally, we will modify slightly the algorithm in \cref{sub:query} using string synchronizing sets. By doing so, we will obtain an algorithm for the same problem requiring $\tilde O (k^{3/4}n^{1/2}m^{o(1)})$ query complexity and $\tilde{O}(k\sqrt{n})$ time complexity.

Throughout this section we will use algorithmic tools and quantum subroutines, which will frequently be very similar to each other. For the sake of convenience, we will describe some of them now. This will allow us to use them as blackboxes in the next subsections:

\begin{lem}
\label{lem:expsearch}
Let $S$ be a string over an alphabet $\Sigma$, let $Q$ denote a string period, and let $k \in [1\dd |S|]$ be a threshold. Then, we can identify the first $k$ mismatches between $S$ and $Q^*$ in $\tilde{O}(\sqrt{k\cdot |S|})$ time.
\end{lem}

\begin{proof}
Let $\sigma_1 < \sigma_2 < \cdots < \sigma_k$ denote the indices of the first $k$ mismatches between $S$ and $Q^*$. Suppose that we have just determined $\sigma_i$ for some $0 \leq i < k$ (if $i = 0$, then $\sigma_0$ is the staring position of the string). To find $\sigma_{i+1}$ we can execute an exponential search combined with Grover search.

More specifically, at the $j$-th jump of the first stage of the exponential search we use Grover search to check if a mismatch exists between $S[\sigma_i, \min(\sigma_i + 2^{j - 1},|S|)]$ and $Q^*[\sigma_i, \min(\sigma_i + 2^{j - 1},|S|)]$. Once we have found the first $j$ for which a mismatch exists, we use a binary search combined with Grover search to find the exact position of $\sigma_{i+1}$. 

By doing so, we use at most $O(\log (\sigma_{i+1} - \sigma_{i})) \leq O(\log |S|)$ Grover searches as subroutine, each of them requiring at most $O(\sqrt{\sigma_{i+1} - \sigma_i})$ time. As a consequence, finding all $\sigma_1, \ldots, \sigma_k$ means spending $\tilde{O}(\sum_{i=0}^{k} \sqrt{\sigma_{i+1} - \sigma_i})$ time. This expression is maximized if all mismatches are equally spaced from each other, meaning that $\sigma_{i+1} - \sigma_{i} = O(|S|/k)$ for all $0 \leq i < k$. Therefore, we can conclude that $\tilde{O}(\sum_{i=0}^{k-1} \sqrt{\sigma_{i+1} - \sigma_i}) \leq \tilde{O}(k \cdot \sqrt{|S|/k}) = \tilde{O}(\sqrt{k\cdot |S|})$. 
\end{proof}

\begin{lem}
\label{lem:candidatepos} %
Let $T \in \Sigma^n$ be a text, let $P \in \Sigma^m$ be a pattern, let $k \in [1\dd n]$ be a threshold, and let $1\leq i \leq n - m$ be a candidate position for a $k$-mismatch occurrence of $P$ in $T$. Then, we can verify if $\delta_H(T[i\dd  i + m], P) \leq k$ in $\tilde{O}(\sqrt{km})$ time.
\end{lem}

\begin{proof}
The proof is almost identical to the proof of \cref{lem:expsearch}.
\end{proof}

\subsection{Finding a Structural Decomposition of the Pattern}
\label{sub:decomp}

Our algorithm relies on a structural lemma proved by Charalampopoulos, Kociumaka and Wellnitz \cite{ckw20} which uses concepts closely related to the notion of periodicity of strings: a string $P$ has \textit{approximate string period} $Q$ if it is identical to the periodic extensions of $Q$ up to a certain number of mismatches (this number will vary depending on the lemma), moreover, if a substring of $P$ has an approximate period, then we call it a \textit{repetitive region}. The structural result is the following:

\begin{lem}[Lemma 3.6 in \cite{ckw20}]
\label{lem:decomp}
Given a string $P \in \Sigma^m$ and a threshold $k \in [1\dd m]$, at least one of the following holds:
\begin{enumerate}
    [label=\textbf{(\alph*)}]
\item The string $P$ contains $2k$ disjoint \textnormal{breaks} $B_1, \ldots, B_{2k}$ each having period $\per(B_i) > m/128k$ and length $|B_i| = \lfloor m/8k \rfloor$. 
\item The string $P$ contains $r$ disjoint \textnormal{repetitive regions} $R_1, \dots, R_r$ of total length $\sum_{i=1}^{r} |R_i| \geq 3/8 \cdot m$ such that each region $|R_i| \geq m/8k$ and has a primitive \textnormal{approximate period} $Q_i$ with $|Q_i| \leq m/128k$ and $\delta_{H}(R_i, Q_i^*) = \lceil 8k/m \cdot |R_i| \rceil$.
\item The string $P$ has a primitive \textnormal{approximate period} $Q$ with $|Q| \leq m/128k$ and $\delta_H(P, Q^*) < 8k$.
\end{enumerate}
\end{lem}

\cite{ckw20} proved this structural lemma in a constructive way: they gave an algorithm for determining either a set $\mathcal{B}$ of $2k$ breaks, or a set $\mathcal{R}$ of repetitive regions of total length at least $3/8 \cdot m$, or an approximate string period $Q$. In this section, we will describe a direct quantum speedup of their algorithm.
 For the sake of convenience, we include here their pseudocode (see Algorithm~\ref{alg:decompalgo}). 

\begin{algorithm2e}
\caption{The constructive proof of \cref{lem:decomp} given by \cite{ckw20}}
\label{alg:decompalgo}

\Begin{
$\mathcal{B} \leftarrow \{\}$\;
$\mathcal{R} \leftarrow \{\}$\;
\While{\textbf{true}}{
    Consider the fragment $P' = P[j\dd j+\lfloor m/8k\rfloor)$ of the next $\lfloor m/8k \rfloor$ unprocessed characters of $P$\;
    \If{$\text{per}(P') > m/128k$}{
        $\mathcal{B} \leftarrow \mathcal{B} \cup \{P'\}$\;
        \If{$|\mathcal{B}| = 2k$}{
            \KwRet{$\mathcal{B}$}
        }
    }\Else{
        $Q \leftarrow P[j\dd j+\per(P'))$\;
        Search for a prefix $R$ of $P[j \dd m)$ with $|R| > |P'|$ and $\delta_H(R, Q^*) = \lfloor 8k/m \cdot |R| \rfloor$\;
        \If{such $R$ exists}{
            $\mathcal{R} \leftarrow \mathcal{R} \cup \{(R, Q)\}$\;
            \If{$\sum_{(R, Q) \in \mathcal{R}} |R| \geq 3/8 \cdot m$}{
                \KwRet{$\mathcal{R}$}}
            } \Else {
                Search for a suffix $R'$ of $P$ with $|R'| \geq m - j$ and $\delta_H(R', \rot^{|R'|-m+j}(Q)^*) = \lceil 8k/m \cdot |R'| \rceil$\;
                \If{such $R'$ exists}{
                    \KwRet {repetitive region $(R', \rot^{|R'|-m+j}(Q))$}
                } \Else {
                    \KwRet{approximate period rot$^j(Q)$}
                }
            
        }
    }
}
}

\end{algorithm2e} 

Here we provide an informal overview of Algorithm~\ref{alg:decompalgo}; its proof of correctness can be found in \cite{ckw20}. 
The algorithm maintains an index $j$ which indicates up until which position the string has been processed, and returns as soon as we have found one of the structural properties described in \cref{lem:decomp}.
 The algorithm is cleverly designed in such a way that it returns always before $j \geq 5/8 \cdot m$ holds.
  In each step, we consider the fragment $P' = P[j \dd j+\lfloor m/8k \rfloor]$. If $\per(P') \geq m/128k$, we can add $P'$ to $\mathcal{B}$.
   Otherwise, we will try to extend $P'$ to a repetitive region. Let $Q$ the string period of $P'$. We extend $P'$ by searching for a prefix of $P[j\dd m)$ longer than $P'$ which has exactly $8k/m \cdot |R|$ mismatches with $Q^*$. If we managed to find such a prefix, we can add $R$ to $\mathcal{R}$. If not, it means that there were not enough mismatches between $Q^*$ and $P[j\dd m)$. At this point, we try to extend $P[j\dd m)$ backwards to a repetitive region. If we manage to do so, we already have a long enough repetitive region to return as $j < 5/8 \cdot m$. Otherwise, it means that the whole string $P$ has few mismatches with some shift of $Q^*$, and thus $P$ has an \textit{approximate period}. A more detailed proof can be found in \cite{ckw20}.

\begin{lem}
\label{lem:quantumdecomp}
Given a string $P$, we can find a structural decomposition as described in \cref{lem:decomp} in $\tilde{O}(\sqrt{km})$ quantum time.
\end{lem}

\begin{proof} 
We show how to implement Algorithm \ref{alg:decompalgo} in $\tilde{O}(\sqrt{km})$ quantum time. There is no need to modify its pseudocode. We will go through line by line, and show how we can turn it into a quantum algorithm by using suitable quantum subroutines.

First, notice that the while loop at line 4 requires at most $O(k)$ iterations, as we process at least $\lfloor m/8k \rfloor$ characters in each iteration. Since it possible to find the period of $P'$ in $\tilde{O}(\sqrt{m/k})$ (\cref{cor:per}), we can conclude that line 5 summed over all iterations runs in $\tilde{O}(k\cdot \sqrt{m/k}) = \tilde{O}(\sqrt{km})$ time.

To find the prefix $R$ at line 12 we use a similar approach as in \cref{lem:expsearch}. The key insight is the following: if the prefix $R$ exists, then its last character is a mismatch with $Q^*$. This motivates us to find one after another the mismatches between $P[j\dd m)$ and $Q^*$ similarly to \cref{lem:expsearch}, and to end prematurely the subroutine if we manage to find such a prefix $R$. To see why this works, consider the difference $\Delta_\rho : = \lceil 8k/m \cdot \rho \rceil - \delta_H(P[j\dd j+\rho], Q^*)$ for $\rho \in [|P'|\dd m-j]$. Every time we encounter a new mismatch at a position $\rho$ we have that $\Delta_{\rho}$ decreases at most by one. For any other position $\Delta_{\rho}$ cannot decrease. As we start with $\Delta_{|P'|} = 1$ and we find the prefix $R$ as soon as $\Delta_\rho =0$, we conclude that the last character of $R$ has to be necessarily a mismatch with $Q^*$. Hence, if there exists such a prefix $R$, we can find it in $\tilde{O}(\sqrt{k/m} \cdot |R|)$ time. Summed over all $R \in \mathcal{R}$ this results in $\sum_{(R, Q)} \tilde{O}(\sqrt{k/m} \cdot |R|) = \tilde{O}(\sqrt{k/m}) \sum_{(R, Q)} |R| \leq \tilde{O}(\sqrt{km})$ time since it holds that $\sum_{(R, Q)} |R| \leq m$. 

Conversely, if we did not manage to find such a prefix $R$, we will return either at line 20 or at line 22. Before returning, we still have to determine if $R'$ exists. Similarly, we can argue that the search for $R'$ will require no more than $\tilde{O}(\sqrt{km})$ time. Lastly,  returning an \textit{approximate period} at line 22 means not having found enough mismatches with $Q^*$ at line 12 and 18; more specifically, we encountered at most $8k$ of them. From \cref{lem:expsearch} follows that the subroutine required $\tilde{O}(\sqrt{km})$ time.
\end{proof}

\subsection{Finding a \texorpdfstring{$k$}{k}-Mismatch Occurrence in \texorpdfstring{$\tilde{O}(k\sqrt{n})$}{tilde{O}(k*sqrt(n))} Time}
\label{sub:matching}

Each of the three cases described in \cref{lem:decomp} gives us useful structure that can be used to find a $k$-mismatch occurrence of $P$ in $T$ in $\tilde{O}(k\sqrt{n})$ time. We consider each case separately in \cref{lem:breakcase}, \cref{lem:repregcase} and \cref{lem:approxcase}.

\begin{lem}
\label{lem:breakcase}
Let $P \in \Sigma^m$, let $T\in \Sigma^n$, and let $k \in [1\dd m]$ be a threshold. Further, assume $P$ contains $2k$ disjoint \textnormal{breaks} $B_1, \ldots, B_{2k}$ each having period $\per(B_i) > m/128k$ and the length $|B_i| = \lfloor m/8k \rfloor$. 

Then, we can verify the existence of a $k$-mismatch occurrence of $P$ in $T$ (and report its starting position in case it exists) in $\tilde{O}(k\sqrt{n})$ time.
\end{lem}

\begin{proof}
We select u.a.r. a \textit{break} $B$ among the $2k$ disjoint \textit{breaks} $B_1, \ldots, B_{2k}$. In a $k$-mismatch occurrence of $P$ in $T$, there are at least $k$ \textit{breaks} among $B_1, \ldots, B_{2k}$ where no mismatch occurs. Therefore, the probability that no mismatch occurs in \textit{break} $B$ is at least $k/2k = 1/2$. Further on, assume that we have selected w.h.p. a break $B$ where no mismatch is located starting at position $\beta$ in the pattern.

Since it holds that $\per(B) > m/128k$ we know that $B$ can match exactly in $T$ at most $|T| / \per(B) \le O(nk/m)$ times. Therefore, if we find the exact matches of $B$ in $T$, we obtain $O(nk/m)$ candidate positions that we can verify using \cref{lem:candidatepos}. 

For that purpose, we partition the text $T$ into blocks of length $|B|/4 \geq m/32k$, obtaining $O(n/|B|) = O(nk/m)$ blocks. For each of these blocks, we search for an exact match of $B$ with starting position in the block in $\tilde{O}(\sqrt{m/k})$ time using \cref{thm:matching}. From $\per(B) > m/128k$ it follows that there are at most $4$ exact matches of $B$ in every block.
 If we have found a perfect match at position $\gamma$, then our candidate position will be $\gamma - \beta$. We can verify a constant number of candidate positions using \cref{lem:candidatepos} in $\tilde{O}(\sqrt{km})$. Therefore, we can use Grover search over all $O(nk/m)$ blocks to find a $k$-mismatch occurrence in $\tilde{O}(\sqrt{nk/m} \cdot (\sqrt{m/k} + \sqrt{km})) = \tilde{O}(k\sqrt{n})$ time.
\end{proof}

In both the second and the third case of \cref{lem:decomp}, we work with strings having an \textit{approximate period}.
Before proceeding to devise algorithms for these two cases, we show how to match efficiently with few mismatches a \textit{repetitive region} $R$ having an \textit{approximate period} in a text $T$.

\begin{lem}
\label{lem:approxmatching}
Let $R$ and $T$ be strings over the alphabet $\Sigma$ such that $|R| \leq |T| \leq 3/2 \cdot |R|$, and let $m, k$ be integers such that $m/8k \leq |R| \leq m$ and $k\leq m$. Moreover, assume $R$ has a primitive \textnormal{approximate period} $Q$ with $|Q| \leq m/128k$ and $\delta_{H}(R, Q^*) \leq \lceil 8k/m \cdot |R| \rceil$. 

Then, for any integer $k' \leq \min(\lfloor 4k/m \cdot |R| \rfloor, k)$ and $\gamma \in [1\dd |T|-|R|]$, after $\tilde{O}(\sqrt{k/m} \cdot |R|)$ time preprocessing, we can verify the existence of a $k'$-mismatch occurrence of $R$ at position $\gamma$ of $T$ in $\tilde{O}(k/m \cdot |R|)$ time.
\end{lem}

\begin{proof}
Intuitively, the existence of a $k'$-mismatch occurrence implies that $T$ has also primitive \textit{approximate period} $Q$ with slightly more mismatches at the position where the $k'$-mismatch occurrence is located. Hence, a premise for a $k'$-mismatch occurrence of $R$ in $T$ is the existence of a shift of $Q^*$ with which a substring of $T$ has few mismatches. We will show that all $k'$-mismatch occurrences of $R$ in $T$ will have to align with this shift of $Q^*$, given it exists. Further on, let $h := \delta_{H}(R, Q^*) \leq \lceil 8k/m \cdot |R| \rceil$. We begin by finding the shift of $Q^*$ with which a substring of $T$ possibly aligns.

For this purpose, we divide $T[|R|/2\dd |R|]$ into segments of length $2|Q|$, obtaining $|R|/4|Q| \geq 32k/m \cdot |R|$ full segments. By selecting u.a.r one of these, and by matching $Q$ exactly in it (we look for occurrences of $Q$ that are fully contained in the segments), we ensure that it will be aligned with the \textit{approximate period} with at least constant probability. By doing so, we might also find an occurrence of $Q$ that is not necessarily aligned with the approximate period, however, our approach ensures that there are enough aligned segments such that we find the right alignment with at least constant probability. More specifically, if there is a $k'$-mismatch occurrence of $R$ in $T$, using the triangular inequality we deduce that there are at most $h + \lfloor 4k/m \cdot |R| \rfloor \leq 2h$ occurrences of $Q$ in $T[|R|/2\dd |R|]$ with a mismatch. This means, there are at least $|R|/4|Q| - 2h \geq 32k/m \cdot |R| - 20k/m \cdot|R| = 12k/m \cdot |R|$ segments where we do not have any mismatch with $Q$. Hence, by selecting u.a.r. one of these we ensure that we find the right alignment with probability at least $12k/m \cdot |R| / (32k/m \cdot |R|) = 3/8$. Using a standard boosting argument we can ensure that this happens w.h.p. Note, we choose the segment size to be $2|Q|$ because if there is no mismatch in a segment, then $Q$ is fully contained in it without errors at least once. This step requires $\tilde{O}(\sqrt{|Q|}) \leq \tilde{O}(\sqrt{m/k}) \leq \tilde{O}(\sqrt{k/m} \cdot |R|)$ as $|Q| \leq m/128k$ and $|R| \geq m/8k$. 

The \textit{approximate period} of any $k'$-mismatch occurrence of $R$ in $T$ is aligned with the occurrence of $Q$ that we have just found. To prove this, it suffices to show that the \textit{approximate periods} of any two $k'$-mismatch occurrence of $R$ in $T$ are aligned. Since we have seen that w.h.p we can select a segment with a correct alignment of $Q$, then there exists such a segment for both occurrences of $R$, hence their approximate periods must be aligned. If we assume that we matched $Q$ exactly at position $\theta$, then our candidate positions for finding a $k'$-mismatch occurrence of $R$ in $T$ are all positions in $T$ which are at distance of a multiple of $Q$ from $\theta$.

Before proceeding to verify whether there is a $k'$-mismatch occurrence of $R$ in $T$ at our candidate positions, we find the first $2h$ mismatches between $Q^*$ on the left and on the right of $\theta$ (we stop respectively at position $\max(\theta-|R|, 1)$ and $\min(\theta+|R|, |T|)$ if there are not enough mismatches). More formally, let $i_L$ be the smallest position of $T$ larger than $\max(\theta-|R|, 1)$  such that $\delta_H(\rot^{\theta - i_L}(Q)^*, T[i_L\dd \theta]) \leq 2h$ and let $i_R$ be the largest position of $T$ smaller than $\min(\theta+|R|, |T|)$ such that $\delta(Q^*, T[\theta, i_R]) \leq 2h$. Note, the hamming distance between $Q^*$ and the substring of $T$ where the $k'$-mismatch occurrence is located is at most $\lfloor 4k/m \cdot |R| \rfloor + h \leq 2h$. Hence, by extending to the right and to the left by $2h$ mismatches we ensure that all $k'$-mismatch occurrences of $R$ in $T$ are contained in $T[i_L, i_R]$.

Now, in order to verify whether there is a $k'$-mismatch occurrence of $R$ in $T$ at candidate position $\gamma$, it suffices to compare those positions in $R$ and $T[i_L\dd i_R]$ that differ from $Q^*$. If we define $M := \{j : Q^*[\gamma+j] \neq T[\gamma+j] \land Q^*[\gamma+j] \neq R[\gamma+j] \land T[\gamma+j] = R[\gamma+j]\}$, then it holds that 
\begin{align*}
\delta_H(T[\gamma\dd \gamma+|R|], R) &= \delta_H(T[\gamma\dd \gamma+|R|], Q^*) + \delta_H(Q^*, R) - |M| \\
&= \delta_H(T[\gamma\dd \gamma+|R|], Q^*) + h - |M|.
\end{align*}

Using \cref{lem:expsearch} we can find the exact positions of the mismatches between $Q^*$ and $T[i_L\dd i_R]$, and between $Q^*$ and $P$ in $\tilde{O}(\sqrt{h \cdot |R|}) = \tilde{O}(\sqrt{k/m} \cdot |R|)$. For each candidate position $\gamma$, we scan simultaneously these mismatches in $O(h)$ to find the exact values of $\delta_H(T[\gamma\dd \gamma+|R|], Q^*)$ and $|M|$. Since the value of $h$ is already known to us, we can verify if $\delta_H(T[\gamma\dd \gamma+|R|], R) \leq k'$ in $O(h) = \tilde{O}(k/m \cdot |R|)$ for each candidate position $\gamma$.
\end{proof}

We also restate here another structural result of \cite{ckw20} that will be crucial to bound the candidate positions in the second case.

\begin{lem}[Corollary 3.5 in \cite{ckw20}]
\label{lem:ckwrep}
Let $P \in \Sigma^m$, let $T \in \Sigma^n$, and let $k \in [1\dd m]$ be a threshold. If there is a positive integer $d \geq 2k$ and a primitive string $Q$ with $|Q| \leq m/8d$ and $\delta_H(P, Q^*) = d$, then the number of $k$-mismatch occurrences of $P$ in $T$ are at most $12 \cdot n/m \cdot d$.
\end{lem}

We have now all the tools to show how to deal with the second case.

\begin{lem}
\label{lem:repregcase}
Let $P \in \Sigma^m$, let $T \in \Sigma^n$, and let $k \in [1\dd m]$ be a threshold. Further, assume $P$ contains $r$ disjoint \textnormal{repetitive regions} $R_1, \dots, R_r$ of total length $\sum_{i=1}^{r} |R_i| \geq 3/8 \cdot m$ such that each region $|R_i| \geq m/8k$ and has a primitive \textnormal{approximate period} $Q_i$ with $|Q_i| \leq m/128k$ and $\delta_{H}(R_i, Q_i^*) = \lceil 8k/m \cdot |R_i| \rceil$. 

Then, we can verify the existence of a $k$-mismatch occurrence of $P$ in $T$ (and report its starting position in case it exists) in $\tilde{O}(k\sqrt{n})$ time.
\end{lem}

\begin{proof}
Let us start to handle this case by selecting a \textit{repetitive region} $R$ among $R_1, \dots, R_r$ with probability $\textnormal{Pr}[R = R_i] = |R_i| / \sum_{(R_i, Q_i) \in \mathcal{R}} |R_i|$, that is, proportionally to their length. 
Then, with at least constant probability in a $k$-mismatch occurrence of $P$ in $T$, out of the at most $k$ mismatches no more than $k' := \lfloor 4k/m \cdot |R| \rfloor$ of them are located in the \textit{repetitive region} $R$. By using a standard boosting argument we ensure that this happens w.h.p. Let $\delta_H^{\pi}(R_i)$ be the number of mismatches located in $R_i$ in the $k$-mismatch occurrence of $P$ in $T$ which starts at position $\pi$. Additionally, we define the index set $I_\pi := \{i : \delta_H^{\pi}(R_i) \leq \lfloor 4k/m \cdot |R_i| \rfloor\}$. For the sake of contradiction, assume $\sum_{i \in I_\pi} |R_i| < m/16$. As a result, $\sum_{i \notin I_\pi} |R_i| = \sum_{R_i \in \mathcal{R}} |R_i| - \sum_{i \in I_\pi} |R_i| \geq 3/8 \cdot m - 1/16 \cdot m \geq 5/16 \cdot m$, and therefore $\delta_H(T[p\dd p+m], P) \geq \sum_{i \notin I_\pi} \delta_H^{p}(R_i) > \sum_{i \notin I_\pi} \lfloor 4k/m \cdot |R_i| \rfloor > k$, which is a contradiction. Hence, we have that $\sum_{i \in I_\pi} |R_i| \geq m/16$. Now, it easy to see that at most $\lfloor 4k/m \cdot |R| \rfloor$ mismatches are located in $R$ with probability $\sum_{i \in I_\pi} |R_i| / \sum_{(R_i, Q_i) \in \mathcal{R}} |R_i| \geq (m/16) / m = 1/16$.

The number of occurrences with $k' := \lfloor 4k/m \cdot |R| \rfloor$ mismatches of $R$ in $T$ can be upper bounded by $192 \cdot nk/m$. To show this, we let $Q$ be the \textit{approximate period} of $R$, and we set $d := \delta_H(R, Q^*)$. By doing so, we obtain $d \geq 2k'$ and $d = \lceil 8k/m \cdot |R| \rceil \leq 16 \cdot k/m \cdot |R|$ since $|R| \geq m/8k$. This implies $|Q| \leq m/128k \leq |R|/8d$. The assumptions of \cref{lem:ckwrep} are satisfied. Therefore, we can upper bound the $k'$-mismatch occurrences of $R$ in $T$ with $12\cdot n/|R| \cdot d \leq 192 \cdot nk/m$ since $d \leq 16 \cdot k/m \cdot |R|$. These occurrences deliver us $\mathcal{O}(nk/m)$ candidate positions as in \cref{lem:breakcase}.

We proceed by partitioning the text $T$ into blocks of length at most $|R|/2$, and use Grover’s algorithm to search for a block which contains the starting position of a $k'$-mismatch occurrence of $R$. Note, this corresponds to searching for a $k'$-mismatch occurrence of $R$ which is completely contained in a substring of $T$ of length $3/2 \cdot |R| - 1$. Fix one of these substrings, and denote it with $T'$. Using \cref{lem:approxmatching} we can preprocess $T'$ in $\tilde{O}(\sqrt{k/m} \cdot |R|)$ time and  verify the existence of a $k'$-mismatch occurrence at any position of $T'$ in $\tilde{O}(k/m \cdot |R|)$ time.

Being able to find all $k'$-mismatch occurrences of $R$ in $T'$, we can now turn our attention to finding a $k$-mismatch occurrence of $P$ in $T$. If $R$ starts in $P$ at position $\rho$, our candidate position for a $k$-mismatch occurrence given a $k'$-mismatch occurrence at position $\gamma$ is position $\gamma - \rho$. However, we have to be very careful when doing this: in each block there might be more than a constant number of $k'$-mismatch occurrences of $R$ in $T'$. This forbids us to find $\textit{before}$ one by one all $k'$-mismatch occurrences of $R$ in $T'$, and to verify $\textit{afterwards}$ all candidate positions which were found for each block as we did in \cref{lem:breakcase}. As retrieving all $k'$-mismatch occurrences might be too expensive, we check for a $k$-mismatch occurrence directly in our subroutine for finding a $k'$-mismatch occurrence only in case it was successful. This complicates the analysis, as there might be subroutines that run longer than other. To solve this issue we will use quantum search with variable times.

For this purpose, let $\Pi_i$ be the running time that we spend on the $i$-th block. Moreover, let $\tau_i$ denote the number of $k'$-mismatch occurrences that we have found in the $i$-th block we partitioned $T$ into. Note, $\tau_i$ is also the number of times we execute the subroutine described in \cref{lem:candidatepos} in the $i$-th block. We have to take into account two terms in $\Pi_i$: preprocessing $T'$ in $\tilde{O}(\sqrt{k/m} \cdot |R|)$ and verifying all candidate positions by performing a quantum search with variable times. For $\tau_i$ candidate positions this takes $\tilde{O}(k/m \cdot |R| + \sqrt{km})$ time, whereas for the remaining ones, which are at most $O(|R|)$, it takes only $O(k/m \cdot |R|)$. As a result, using quantum search with variable times we have that
\begin{align*}
\Pi_i &\in \tilde{O}\left(\sqrt{k/m} \cdot |R| + \sqrt{|R| \cdot (k/m \cdot |R|)^2 + \tau_i \cdot (k/m \cdot |R| + \sqrt{km})^2}\right)   \\
&\leq \tilde{O}\left(\sqrt{k/m} \cdot |R| + \sqrt{|R| \cdot (k/m \cdot |R|)^2 + \tau_i \cdot km}\right),
\end{align*}
where the last equation follows from the fact that $(a + b)^2 \in O(a^2 + b^2)$ and $\tau_i \leq |T_i| \leq 3/2 \cdot |R|$. Now, we use again the quantum search with variable times over all blocks of length $|R|/2$ that we divided $T$ into. The overall running time $\Pi$ of the quantum algorithm becomes
\begin{align*}
\Pi \in \tilde{O}\left(\sqrt{\sum_i \Pi_i^2}\right) &\leq  \tilde{O}\left(\sqrt{\sum_{i} \left(\sqrt{k/m} \cdot |R| + \sqrt{|R| \cdot (k/m \cdot |R|)^2 + \tau_i \cdot km}\right)^2}\right) \\
&\leq \tilde{O}\left(\sqrt{\sum_{i} \left(k/m \cdot |R|^2 + |R| \cdot (k/m \cdot |R|)^2 + \tau_i \cdot km\right)}\right)\\
&\leq \tilde{O}\left(\sqrt{n/|R| \cdot k/m \cdot |R|^2 + n/|R| \cdot |R| \cdot (k/m \cdot |R|)^2 + \sum_i \tau_i \cdot km}\right)\\
&\leq  \tilde{O}\left(\sqrt{n} \cdot \sqrt{k/m \cdot |R|} + \sqrt{n} \cdot (k/m \cdot |R|) + \sqrt{\sum_i \tau_i} \cdot \sqrt{km}\right)\\
&\leq  \tilde{O}(\sqrt{kn} + k\sqrt{n} + \sqrt{nk/m} \cdot \sqrt{km}) \leq  \tilde{O}(k\sqrt{n}),
\end{align*}
where we have used that $\sum_i \tau_i \in O(nk/m)$. This concludes the proof.
\end{proof}

\begin{lem}
\label{lem:approxcase}
Let $P \in \Sigma^m$, let $T\in \Sigma^n$, and let $k \in [1\dd m]$ be a threshold. Further, assume $P$ has a primitive \textnormal{approximate period} $Q$ with $|Q| \leq m/128k$ and $\delta_H(P, Q^*) < 8k$. 

Then, we can verify the existence of a $k$-mismatch occurrence of $P$ in $T$ (and report its starting position in case it exists) in $\tilde{O}(k\sqrt{n})$ time.
\end{lem}

\begin{proof}
Note, the structural properties of $P$ in this case resemble the properties that we assumed for $R$ in the proof of \cref{lem:repregcase}. This allows us to use again \cref{lem:approxmatching}.

We partition the text $T$ into blocks of length $m/2$, and use Grover’s algorithm to search for a block which contains the starting position of a $k'$-mismatch occurrence of $R$. Note, this corresponds to searching for a $k$-mismatch occurrence of $P$ which is completely contained in a substring of $T$ of length $3/2 \cdot m - 1$. Let us fix one of these substrings, and let us denote it with $T'$. Using \cref{lem:approxmatching} we can preprocess $T'$ in $\tilde{O}(\sqrt{km})$ time and we can verify the existence of a $k$-mismatch occurrence at any position of $T'$ in $\tilde{O}(k)$ time. As $|T'| \in O(m)$, we can use Grover's algorithm to search for a $k$-mismatch occurrence in $T'$ in $\tilde{O}(k\sqrt{m})$. By performing another Grover search over all blocks we can find a $k$-mismatch occurrence of $P$ in $T$ in $\tilde{O}(\sqrt{n/m} \cdot (\sqrt{km} + \sqrt{m} \cdot k)) = O(k\sqrt{n})$ time.
\end{proof}

We have now covered all three cases. From these \cref{thm:mismatch} follows immediately.

    \begin{proof}[Proof of \cref{thm:mismatch}]
The first step of the algorithm consists into computing the decomposition of the string as described in \cref{lem:quantumdecomp} in $\tilde{O}(\sqrt{km})$ time. We then perform a case distinction depending on which case of \cref{lem:decomp} occurred. \cref{lem:breakcase}, \cref{lem:repregcase}, and \cref{lem:approxcase} show how to devise algorithms that verify the existence of a $k$-mismatch occurrence of $P$ in $T$ (and report its starting position in case it exists) in $\tilde{O}(k\sqrt{n})$ time for each of these cases.
\end{proof}

\subsection{Improvement via String Synchronizing Sets}
\label{sub:query}
In this subsection, we show that it is possible to slightly modify the algorithm in order to improve its query complexity (up to an $m^{o(1)}$ factor). The first step towards an improvement consist into adapting \cref{lem:candidatepos} aiming for a better query complexity. We do this by using \cref{thm:candidatepos-sync}. Another adjustment is needed when matching with few errors a string with an \textit{approximate period} in the text. For this purpose we reformulate \cref{lem:approxmatching} obtaining \cref{lem:approxmatching2}.

\begin{lem}
\label{lem:approxmatching2}
Let $R$ and $T$ be strings over the alphabet $\Sigma$ such that $|R| \leq |T| \leq 3/2 \cdot |R|$, and let $m, k$ be integers such that $m/8k \leq |R| \leq m$ and $k\leq m$. Moreover, assume $R$ has a primitive \textnormal{approximate period} $Q$ with $|Q| \leq m/128k$ and $\delta_{H}(R, Q^*) \leq \lceil 8k/m \cdot |R| \rceil$. 

Then, for any integer $k' \leq \min(\lfloor 4k/m \cdot |R| \rfloor, k)$ and $\gamma \in [1\dd |T|-|R|]$, after $\tilde{O}(\sqrt{k/m} \cdot |R|)$ query and time preprocessing, we can verify the existence of a $k'$-mismatch occurrence of $R$ at position $\gamma$ of $T$ in $\tilde{O}(k/m \cdot |R|)$ time. 

Additionally, one of the two following holds:
\begin{itemize}
    \item Verifying the existence of a $k'$ mismatch at position $\gamma$ does not require any further query complexity.  
    \item The number of candidate positions of $k'$-mismatch occurrences of $R$ in $T$ are bound by $3/2 \cdot m/k$, each of them can be verified in $\tilde{O}(k/m \cdot |R|)$ query complexity.
\end{itemize}
\end{lem}

\begin{proof}
The proof is almost identical to the proof of \cref{lem:approxmatching}, therefore we will only highlight the differences. Recall, if we assume that we match $Q$ exactly at position $\theta$, then our candidate positions for finding a $k'$-mismatch occurrence of $R$ in $T$ are all positions in $T$ which are at distance of a multiple of $Q$ from $\theta$. This fact leads us to perform a case distinction depending on the length of the \textit{approximate period} $Q$. If $|Q| \geq k/m \cdot |R|$, then we have at most $|T|/|Q| \leq 3/2 \cdot |R| / (k/m \cdot |R|) = 3/2 \cdot m/k$ candidate positions, each of them can be verified in $\tilde{O}(k/m \cdot |R|)$ query and time complexity using the same approach as in the proof of \cref{lem:approxmatching}. Conversely, if $|Q| < k/m \cdot |R|$, then we can read all characters of $Q$ one by one requiring $\tilde{O}(k/m \cdot |R|) \leq \tilde{O}(\sqrt{k/m} \cdot |R|)$ query complexity. As a result, we know all characters of $R$ and $T[i_L\dd i_R]$ and we can look for all $k'$-mismatch occurrence of $R$ in $T[i_L\dd i_R]$ without spending any additional query complexity. 
\end{proof}

Note, every time we use \cref{lem:approxmatching2} in our query complexity analysis we will assume that the second case occurs as we consider only the worst case. Also, notice that the time complexity of verifying a position remained the same as in \cref{lem:approxmatching}. We now prove the main result of this section.

\begin{proof} [Proof of \cref{thm:querymismatch}]
We proceed by partitioning the text $T$ into $O(n/m)$ blocks of length $m/2$, and use Grover’s algorithm to search for a block which contains the starting position of a $k$-mismatch occurrence of $T$. Note, this corresponds to searching for a $k$-mismatch occurrence of $T$ which is completely contained in a substring of $T$ of length at most $3/2 \cdot m - 1$. Now, fix one of these blocks and denote it with $T'$. If we show that we can find a $k$-mismatch occurrence of $P$ in $T'$ in $\tilde O (k^{3/4}m^{1/2 + o(1)})$ query complexity and $\tilde O (k\sqrt{m})$ time complexity, finding a $k$-mismatch occurrence of $P$ in $T$ requires $\tilde O (\sqrt{n/m} \cdot k^{3/4}m^{1/2 + o(1)}) = \tilde O (k^{3/4}n^{1/2}m^{o(1)})$ query complexity and $\tilde O (\sqrt{n/m} \cdot k\sqrt{m}) = \tilde O (k\sqrt{n})$ time complexity.

First, we consider the case $m < k^{3/2}$. We read all characters of $T'$ and $P$. This requires $O(m) \leq \tilde O (k^{3/4}\sqrt{m})$ query and time complexity. Now, we use \cref{thm:mismatch} to find a $k$-mismatch occurrence of $P$ in $T'$. Note, this will not cost us any further query complexity as we execute the algorithm on the string we just read.

Otherwise, if $m \geq k^{3/2}$, we perform the preprocessing described in \cref{lem:quantumdecomp} in $\tilde{O}(\sqrt{km})$ time, and we build the data structure of \cref{thm:candidatepos-sync} in $\tilde{O}(k^{3/4}m^{1/2+o(1)})$ query and time complexity by setting $\tau = m/k^{3/2}$. Note, this allows us also to verify a candidate position in $\tilde{O} (k^{1/4}\sqrt{m})$ query complexity. Next, we perform a case distinction depending on which case of \cref{lem:decomp} occurred:

\begin{enumerate}[label=\textbf{(\alph*)}]
\item If $|P|$ contains $2k$ disjoint \textnormal{breaks} $B_1, \ldots, B_{2k}$ each having period $\per(B_i) > m/128k$ and the length $|B_i| = \lfloor m/8k \rfloor$, we use the same approach as in \cref{lem:breakcase}. This time, we verify all candidate position taking $\tilde{O}(\sqrt{k} \cdot k^{1/4}m^{1/2}) = \tilde{O}(k^{3/4}m^{1/2})$ query and time complexity. Note, this time the candidate position are $O(k)$ and not $O(nk/m)$ anymore, since we are working with $T'$ and not with $T$. The same holds also for the second case.

\item Next, we consider the case where $P$ contains $r$ disjoint \textnormal{repetitive regions} $R_1, \dots, R_r$ of total length $\sum_{i=1}^{r} |R_i| \geq 3/8 \cdot m$ such that each region $|R_i| \geq m/8k$ and has a primitive \textnormal{approximate period} $Q_i$ with $|Q_i| \leq m/128k$ and $\delta_{H}(R_i, Q_i^*) = \lceil 8k/m \cdot |R_i| \rceil$. We use the same algorithm as described in the proof of \cref{lem:approxcase}, replacing \cref{lem:approxmatching} and \cref{lem:candidatepos} with respectively \cref{lem:approxmatching2} and \cref{thm:candidatepos-sync}. Let us analyze the query complexity first, which is again tricky. Let $\Pi_i'$ be the query complexity that we spend on the $i$-th block, and let $\tau_i$ denote the number of $k'$-mismatch occurrences that we have found in the $i$-th block we partitioned $T'$ into. Then, we can upperbound $\Pi_i'$ in both cases of \cref{lem:approxmatching2} with 
\begin{align*}
\Pi_i' &\in \tilde{O}\left(\sqrt{k/m} \cdot |R| + \sqrt{m/k \cdot (k/m \cdot |R|)^2 + \tau_i \cdot (k/m \cdot |R| + k^{1/4}m^{1/2})^2}\right)   \\
&\leq \tilde{O}\left(\sqrt{k/m} \cdot |R| + \sqrt{k/m \cdot |R|^2 + \tau_i \cdot \sqrt{k}m}\right)\\
&\leq \tilde{O}\left(\sqrt{k/m} \cdot |R| + \sqrt{\tau_i}k^{1/4}m^{1/2}\right).
\end{align*}

The overall query complexity $\Pi'$ for finding a $k$-mismatch occurrence of $P$ in $T$ becomes
\begin{align*}
\Pi' \in \tilde{O}\left(\sqrt{\sum_i \Pi_i'^2}\right) &\leq  \tilde{O}\left(\sqrt{\sum_{i} \left(\sqrt{k/m} \cdot |R| + \sqrt{\tau_i}k^{1/4}m^{1/2}\right)^2}\right) \\
&\leq \tilde{O}\left(\sqrt{\sum_{i} \left(k/m \cdot |R|^2 + \tau_i \cdot \sqrt{k}m\right)}\right)\\
&\leq \tilde{O}\left(\sqrt{ m/|R| \cdot k/m \cdot |R|^2 + \sum_i \tau_i \cdot \sqrt{k}m}\right)\\
\end{align*}
\begin{align*}
\hspace{3cm}&\leq  \tilde{O}\left(\sqrt{k \cdot |R|} + \sqrt{\sum_i \tau_i} \cdot k^{1/4}m^{1/2}\right)\\
&\leq  \tilde{O}(\sqrt{km} + \sqrt{k} \cdot k^{1/4}m^{1/2}) \leq  \tilde{O}(k^{3/4}m^{1/2}).
\end{align*}

Similarly, we can show that the time complexity is $\tilde O (k\sqrt{m} + k^{3/4}m^{1/2})$. The calculation are exactly the same as in \cref{lem:repregcase} except for the fact that verifying a position costs $k^{3/4}m^{1/2}$ and not $\sqrt{km}$.

\item Lastly, if $|P|$ has a primitive \textit{approximate period} $Q$ with $|Q| \leq m/128k$ and $\delta_H(P,Q^*) < 8k$, then, we use the same approach as in \cref{lem:approxcase}. This time, we use \cref{lem:approxmatching2} instead of \cref{lem:approxmatching}. By doing so we obtain an algorithm requiring $\tilde{O}(\sqrt{km})$ query complexity and $\tilde O(k\sqrt{m})$ time complexity. Note, this time we only need $\tilde{O}(\sqrt{km})$ query complexity, since in the worst case we have to verify $O(m/k)$ positions requiring $\tilde{O}(\sqrt{m/k} \cdot k) = \tilde{O}(\sqrt{km})$ query complexity.

\end{enumerate}
Hence, we have shown how in all cases when $m \geq k^{3/2}$ we can find a $k$-mismatch of $P$ in $T'$ requiring no more than $\tilde O (k^{3/4}m^{1/2 + o(1)})$ query complexity and $\tilde O (k\sqrt{m})$ time complexity.
\end{proof}

\section{Open Problems}
\label{sec:open}
We conclude by mentioning several open questions related to our work.
\begin{itemize}
	\item Can we improve the extra $\tau^{o(1)}$ factors in the sparsity and the time complexity of our string synchronizing set to poly-logarithmic?
	\item Can we improve the quantum query complexity of the $k$-mismatch matching algorithm to closer to the lower bound $\sqrt{kn}$?
	\item String synchronizing sets have had many applications in classical string algorithms. Can our new result find more applications in quantum string algorithms?  An interesting question is the quantum query complexity for computing  the edit distance of two strings, provided that it is no larger than $k$.
\end{itemize}

\section*{Acknowledgement} We thank Virginia Vassilevska Williams for many useful discussions.

	\bibliographystyle{alphaurl} 
	\bibliography{main}

	\appendix
\section{The Deterministic Sampling Technique}
\label{app:ds}
In this section we include the proof of \cref{lem:ds} (restated below) given by Wang and Ying \cite{ying}, based on a modification of Vishkin's deterministic sampling \cite{DBLP:journals/siamcomp/Vishkin91,rameshvinay}. \cite{ying} also improved the logarithmic factors in the complexity, which we do not focus on here.
\lemdsrestate*

\begin{proof}
    \begin{algorithm2e}
\caption{Constructing deterministic samples given $S\in \Sigma^n$ and $\seed\in \{0,1\}^w$ }
\label{alg:ds}
\Begin{
Initialize all offsets in $[0\dd \lfloor n/2\rfloor ]$ as \emph{alive}\;
Initialize empty checkpoint list $C$, and $\mathsf{flag}\gets \textbf{false}$\;
\While{\textbf{true}}{
    $p,q \gets $ minimum and maximum offsets alive, respectively \label{line:find}\;
    Find the minimum $i \in  [1+p\dd n+q]$ such that $S[i-p]\neq S[i-q]$ \label{line:pq}\;
    \If{such $i$ exists \label{line:if}}{
        Append $i$ to list $C$, and return ``failure'' if $\mathsf{length}(C)>w$\;
        Let $c_i \gets \begin{cases}
            S[i-p], & \text{if }\seed[\mathsf{length}(C)]=0, \\ S[i-q], & \text{if }\seed[\mathsf{length}(C)]=1; \end{cases}$

             Kill all offsets $\delta'$ with $ S[i-\delta']\neq c_i$ \label{line:kill}\;
    }\ElseIf{$p=q$ or $\mathsf{flag}=\textbf{true}$ \label{line:ei}}{
            \KwRet{offset $\delta:=p$ with checkpoint list $C$ \label{line:ret}}
        }\Else {
            $\mathsf{flag}\gets \textbf{true}$\;
            Kill offset $q$\label{line:kill2}\;
        }
}
}
\end{algorithm2e} 

The algorithm is summarized in Algorithm~\ref{alg:ds}. It terminates in no more than $w+1=O(\log n)$ rounds.  To implement this algorithm, observe that: whether an offset is killed or not can be decided by checking against all checkpoints currently in list $C$ (Line~\ref{line:kill}) (and additionally Line~\ref{line:kill2} if $\mathsf{flag}=\textbf{true}$), in $\polylog(n)$ time. Hence, Line~\ref{line:find} can be implemented using quantum minimum finding, in $\tilde O(\sqrt{n})$ time. Line~\ref{line:pq} can be implemented using Grover search with binary search. Hence, the overall quantum time complexity is $\tilde O(\sqrt{n})$.

Observe that, at Line~\ref{line:kill}, over the random coin $\sigma[\mathsf{length}(C)]\in \{0,1\}$, in expectation at least half of the  offsets that are currently alive will be killed by the current checkpoint.  Hence, after finding $w$ checkpoints, the expected number of offsets that are still alive is at most $\lfloor n/2\rfloor / 2^{w}$.  When there is only one offset alive, Line~\ref{line:ei} successfully terminates. So by Markov's inequality, it reports ``failure'' with no more than $n/2^w$ probability.

It remains to show that the final property is satisfied if the algorithm successfully terminates.

First, notice that whenever an offset $\delta'$ is killed by Line~\ref{line:kill}, it holds for the checkpoint $i$ that $S[i-\delta']\neq c_i$, while $S[i-\delta'']=c_i$ holds for all the offsets $\delta''$ that are still alive, in particular for the offset $\delta$ that will be returned finally. Hence, $S[i-\delta']\neq S[i-\delta]$ for this checkpoint $i$.

Next, observe that whenever an offset $\delta'$ is killed by Line~\ref{line:kill}, all other offsets $\delta'$ with $\per(S)\mid (\delta - \delta')$ are also killed at this point. So we only need to consider congruence class of offsets modulo $\per(S)$. Note that the \textbf{if} check at Line~\ref{line:if} fails if and only if $\per(S) \mid (q-p)$.

If Line~\ref{line:ei} is reached with $\mathsf{flag} = \mathbf{false}$, then $p=q=\delta$, and all other offsets $\delta'\neq \delta$ must have been killed by Line~\ref{line:kill}, so the desired property is satisfied.

Now consider the case where $\mathsf{flag}$ is set $\mathbf{true}$ at some point. Denote the current values of $p,q$ by $p_1,q_1$, which belong to the same congruence class.  Then, consider the next time where the \textbf{if} check at Line~\ref{line:if} fails, at which point the values of $p,q$ are denoted $p_2,q_2$, which belong to the same congruence class. We know $p_1\le p_2 \le q_2 < q_1$. Hence, the congruence classes of offsets in the interval $[p_1\dd p_2)$ have been killed by Line~\ref{line:kill}, and the congruence classes of offsets in the interval $[q_2\dd q_1)$ have also been killed by Line~\ref{line:kill}. This implies that $\bar p_2= \bar q_2 \in \Z_{\per(S)}$ is the only congruence class that is alive. Hence, the desired property is satisfied.
\end{proof}

\end{document}